\documentclass[11pt]{article}
\usepackage{dsfont,amsmath,amsthm,amssymb}
\usepackage{subfigure}
\usepackage{tikz}
\usetikzlibrary{positioning}

\usetikzlibrary{shapes.symbols,arrows}
\usetikzlibrary{matrix,backgrounds}

\usetikzlibrary{shapes.symbols,arrows}
\usepackage{enumerate}
\usepackage{fullpage}

\newcommand{\diam}{\delta}
\newcommand{\bottle}{{b}}
\newcommand{\diamnorm}{\delta_*}

\newcommand{\ka}{{\tau_{\mathrm{e}}}}
\newcommand{\DS}{\beta_{_{DS}}}

\newcommand{\an}{\beta_{_{a}}}
\newcommand{\bcb}{\beta_{_{b}}}
\newcommand{\kDS}{{\tilde \kappa}}

\newcommand{\kbcb}{\kappa}
\newcommand{\dmax}{d_{_{\max}}}
\newcommand{\dmin}{d_{_{\min}}}
\newcommand{\In}{{\mathrm {In}}}
\newcommand{\Out}{{\mathrm {Out}}}

\newcommand{\len}[1]{| \gamma |_{_{#1}}}
\newcommand{\lenij}[1]{| \gamma_{i\, j} |_{_{#1}}}

\newcommand{\IR}{\mathds{R}}

\newcommand{\dG}{\mathds{G}}

\renewcommand{\leq}{\leqslant}

\renewcommand{\geq}{\geqslant}

\newcommand{\one}{\mathbf{1}}
\newcommand{\zero}{\mathbf{0}}

\newtheorem{thm}{Theorem}
\newtheorem{prop}[thm]{Proposition}
\newtheorem{lem}[thm]{Lemma}
\newtheorem{cor}[thm]{Corollary}

\newcommand{\N}{{\mathrm{N}}}
\newcommand{\V}{{\mathcal{V}}}
\newcommand{\W}{{\mathcal{W}}}
\newcommand{\Q}{\mathcal{Q}}

\newcommand{\drawloop}[1]{\draw[node_arrow] (#1) [loop above, looseness=20, out=70,in=110] to (#1);}
\newcommand{\drawarrows}[2]{
  \draw[node_arrow] (#2) [bend left=30] to (#1);
  \draw[node_arrow] (#1) [bend left=30] to (#2);
}

\title{Geometric Bounds for Convergence Rates of  Averaging Algorithms}

 \author{%
 Bernadette Charron-Bost}
 \date{ CNRS, \'Ecole polytechnique, 91128 Palaiseau, France\\~\\
 \today
  }

\begin{document}
\maketitle

~\vspace{-0.7cm}
\begin{abstract}
We develop a generic method for bounding the convergence rate  of an averaging algorithm 
	running in a multi-agent system with a time-varying network, where  the associated 
	stochastic matrices have a time-independent Perron vector.
This method provides bounds on convergence rates that  unify and refine most of the previously known bounds.
They  depend on geometric parameters of the dynamic communication graph 
	such as the normalized diameter or the bottleneck measure.
	
As corollaries of these geometric bounds, we show that the convergence rate of the Metropolis algorithm in a system of 
	$n$ agents is less than $1-1/4n^2$ with any communication graph that may vary in time, but is permanently 
	connected and  bidirectional.
We prove a similar upper bound for the EqualNeighbor algorithm under the additional assumptions that the number
	of neighbors of each agent is constant and that the communication graph is not too irregular.
% some additional assumptions on degree fluctuations over time and space. % the number 
	%of neighbors of each agent is   constant and 
Moreover our bounds offer improved convergence rates for several averaging algorithms and specific families of communication 
	graphs.	
  
Finally  we extend our methodology to a time-varying Perron vector and show
	how  convergence times may dramatically degrade  with even limited 
     variations of Perron vectors.
\end{abstract}
~\vspace{-0.7cm}

\section{Introduction}

Motivated by the  applications of the Internet and the development of mobile devices with communication 
	capabilities, the design of distributed algorithms for networks with a swarm of agents and time-varying connectivity 
	has been the subject of much recent work. 
The algorithms implemented in such dynamic networks ought to be decentralized, using local information, 
	and resilient to mobility and link failures while remaining efficient.

One of the basic problems arising in multi-agent networked systems is an agreement problem, 
	called  \emph{asymptotic  consensus}, or just \emph{consensus}, in which agents are required to  
	compute values that become infinitely close to each other.
For example, in clock synchronization, agents attempt  to maintain a common time scale; 
	or sensors may try to agree on estimates of a certain variable; 
	or vehicles may attempt to align their direction of motions with their neighbors in coordination of UAV's and
	control formation. 
	
\subsection{Network model and averaging algorithms}	

Let us consider a fixed set of agents that operate synchronously and communicate by exchanging values
	 over an underlying time-varying communication network.
In the  consensus problem, the objective is to design distributed algorithms in which the agents 
	start with different initial values and reach agreement on one value that lies in the range of the initial values.
The term of {\em constrained  consensus}  is used when the goal is to compute a specific value in this range 
	(e.g., the average of the initial values).
	
Natural candidates for solving  the consensus problem are the 
	\emph{averaging algorithms} in which each agent maintains a  scalar variable that it
	repeatedly updates to a  convex combination of its own value and of the values it has just received from its neighbors.
The weights used by an agent can only depend on local informations  available to this agent. 
	%(e.g., the current number of $i$'s neighbors).
The matrix formed with the weights at each time step of an averaging algorithm is a stochastic matrix, 
	and  the graph associated to the stochastic matrix coincides with the communication graph. 
Hence, in the discrete-time model, every execution of an averaging algorithm determines a sequence of stochastic matrices.

Every averaging algorithm corresponds to a specific rule for computing the weights.
Three averaging algorithms are of particular interest, namely the \emph{EqualNeighbor} algorithm with weights equal to
	the inverse of the degrees in the communication graph, its space-symmetric version called \emph{Metropolis}, and the 
	\emph{FixedWeight} algorithm  which is a time-uniformization of the  EqualNeighbor algorithm
	in the sense that each agent uses some bound on its degree instead of its (possibly  time-varying) degree.
A specific feature of the Metropolis algorithm is to address the constrained consensus problem with  convergence on
	 the average of the initial values.

The convergence of averaging algorithms has been proved under various assumptions on the connectivity of the 
	communication graph, in particular when it is time-varying but permanently connected~\cite{Mor05,CMA08a}.
The goal in this paper is to establish novel and tight bounds on the convergence rates of averaging algorithms 
	that depend on geometric parameters of the communication graph.
As demonstrated in the simple case of a fixed communication graph and fixed weights, the convergence rate involves 
	the second largest singular values of the corresponding stochastic matrices.
Thus a primary step is to develop geometric bounds of these singular values and to get some control on the successive associated 
	eigenspaces.

\subsection{Contribution}
In this paper, our first contribution concerns upper bounds on the second largest eigenvalue  of a reversible
	stochastic matrix.
We start with an analytic bound and then develop a geometric bound.
% which depends on the \emph{normalized diameter}
%	of associated graphs. 
This second bound compares well with previous geometric bounds derived through Cheeger-like inequalities 
	or Poincar\'e inequalities, and is often much easier to compute.
We derive geometric bounds on the second largest singular value of a reversible stochastic matrix.
We also obtain an analytic bound on the second largest singular value that is weaker, but still holds 
	when the matrix is not reversible.

Our second contribution is a generic method for bounding the convergence rate of an execution of an averaging algorithm 
	when  the associated stochastic matrices have all the same Perron vector.
Combined with the above bounds on the second largest singular value of stochastic matrices, this method provides
	 bounds on convergence rates that 
	 unify and refine most of the previously known bounds.
Basically, the approach consists in masking time fluctuations of the network topology  by a constant Perron vector.
Two typical examples implementing this strategy for coping with time-varying topologies are the \emph{Metropolis} algorithm
	and the \emph{FixedWeight} algorithm.
Using the geometric bounds developed herein, our method offers improved convergence rates of these algorithms for
	large classes of communication graphs.

We show that for any time-varying topology that is permanently connected and bidirectional,
	the convergence rate of the Metropolis algorithm is at most  $1-1/4n^{2} $, where $n$ is the number of agents.
As a byproduct, we obtain that  the second largest eigenvalue of the random walk on a connected regular bidirectional graph
	is in $1-O(n^{-2})$.
A similar result holds for the EqualNeighbor algorithm with limited degree fluctuations over both  time and space:
	the convergence rate is less than $1-1/\,(3 + \dmax - \dmin) \, n^2$ if each agent has a constant number of neighbors
	in the range $[\dmin,\dmax]$.
These two quadratic bounds exemplify the performance of the Poincar\'e inequality developed by Diaconis and Stroock~\cite{DS91}.

Finally, we extend our methodology to a time-varying Perron vector:
	we provide a heuristic analysis of the convergence rates of averaging algorithms that demonstrates 
	how time-fluctuations of Perron vectors
	 may lead to exponential degradation of convergence times.
Our approach consists in replacing the Euclidean norm associated to the Perron vector by the generic semi-norm 
	$\N(x) = \max (x_i) - \min (x_i) $
	defined on $\IR^n$, which  does not depend on Perron vectors anymore.
	
\paragraph{Related work.}

Several geometric bounds on the second largest eigenvalue and the second largest singular value of a reversible 
	stochastic matrix have been previously developed  (e.g., see~\cite{SJ89,Sin92,DS91,Lub89}).
Our geometric bound expressed in terms of the \emph{normalized diameter} of the associated graph is novel to the best of our knowledge.
The analytic bound presented in this paper is a generalization of the bound developed by Nedi\'c et al. for doubly stochastic 
	matrices~\cite{NOOT09}.

Concerning the convergence rate of averaging algorithms, there is also considerable literature.
Let us cite the bound established by Xiao and Boyd for the Metropolis algorithm  on a fixed topology~\cite{XB04},
	the one developed by Cucker and Smale for modelling formation of flocks in a complete graph~\cite{CS07},
	the bound by Olshevsky and Tsitsiklis which concerns the EqualNeighbor algorithm with fixed degrees~\cite{OT11,OT13},
	the analytic bound developed by Nedi\'c et al.~\cite{NOOT09}  in the case of doubly stochastic matrices (and hence, 
	with the typical application to the Metropolis algorithm), 
	and the one developed by Chazelle~\cite{Cha11}  for the FixedWeight algorithm.
	
From the quadratic bound on the hitting time of Metropolis walks established by Nonaka et al.~\cite{NOSY10}, 
	Olshevsky~\cite{Ols17} deduced that the convergence rate of the \emph{Lazy Metropolis} algorithm 	
	in any system of $n$ agents connected by a fixed bidirectional communication graph
	is less than $1-1/71 n^{2} $.
Our general quadratic bound for the Metropolis algorithm is obtained with a different approach
	based on the discrete analog of the Poincar\'e inequality developed by Diaconis and Strook~\cite{DS91}.
Applied to Lazy Metropolis, our approach  gives the improved bound of $1-1/8 n^{2} $.
It also  proves that the 
	quadratic time complexity result in~\cite{Ols17} extends to the case of  time-varying topologies.

The case of time-varying Perron vectors is addressed by Nedi\'c and Liu~\cite{NL17} with a  different 
	method than ours: instead of dealing with the sequence of Perron vectors 
	 and  using the non-Euclidean norm $\N$, they consider 
	the \emph{absolute probability sequence} associated with the sequence 
	of stochastic matrices~\cite{Kol36} and the sequence of associated Euclidean norms.
%In the case of a constant Perron vector,  each vector of the absolute probability sequence actually coincides with this vector.

\section{Preliminaries on stochastic matrices}

\subsection{Notation}

Let $n$ be a positive integer and  let $[n] = \{1,\dots,n \}$.
For every  positive probability vector $\pi \in \IR^n$, we define
	$$<x,y\!>_{\pi} = \sum_{i\in [n]} \pi_i \, x_i \, y_i , $$
	that is a positive definite inner product on $\IR^n$.
The associated Euclidean norm is denoted by~$ \lVert . \lVert_{\pi}$.

For any $n\times n$ square matrix $P$,  $P^{\dag_{\pi}}$ denotes the adjoint of $P$
	with respect to the inner product $<.\, ,.>_{\pi}$.
We easily check that 
	$$P^{\dag_{\pi}}_{i j} = \frac{\pi_j}{\pi_i} P_{j i} .$$ 
Equivalently, 
	$$ P^{\dag_{\pi}} = \delta_{\pi}^{-1} P^{\mathrm{T}} \delta_{\pi} $$
	where $\delta_{\pi}= \rm{diag} (\pi_1,\dots, \pi_n)$ and $P^{\mathrm{T}}$ is $P$'s transpose.

The real vector space generated by $\one = (1, \dots, 1)^T$ is denoted by $\Delta = \IR . \one$, and
	$\Delta^{\bot_{\pi}}$ is the orthogonal complement of  $\Delta $ in $\IR^n$ for the inner product $<.,.>_{\pi}$.
Clearly,  $ \lVert \one \lVert_{\pi} = 1$.

Another norm on $\Delta^{\bot_{\pi}}$ is provided by the restriction to $\Delta^{\bot_{\pi}}$  of the semi-norm $\N$  on $\IR^n$ defined by 
	$$\N(x) = \max_{i \in [n] } (x_i) - \min_{i \in [n] }(x_i) . $$

\subsection{Reversible stochastic matrices}

Let $P$ be a stochastic matrix of size $n$, and let $G_{\! _P}$   denote the directed graph associated to~$P$.
We assume throughout that $P$ is {\em irreducible}, i.e.,  $G_{\! _P}$ is strongly connected.
The Perron-Frobenius theorem shows that  the spectral radius of $P$, namely 1, is an eigenvalue of $P$
	of geometric multiplicity one.
Then $P$ has a unique Perron vector, that is, there is a unique positive probability vector~$ \pi_{_P }$ such that 
	$P^{\mathrm{T}} \, \pi_{_P } = \pi_{_P } $.
The matrix  $P^{\dag_{\pi_{_P}}}$, simply denoted $ P^{\dag}$, is stochastic.
Indeed, $$ \left (  \delta_{\pi_{_P }}^{-1} P^{\mathrm{T}} \delta_{\pi_{_P }} \right )  \one = 
                                              \left (  \delta_{\pi_{_P }}^{-1} P^{\mathrm{T}} \right )  \pi_{_P } =  \delta_{\pi_{_P }}^{-1}  \pi_{_P } = \one.$$
Therefore,  $ \Delta^{\bot_{\pi_{_P }}} $, denoted  $ \Delta^{\bot_{P}} $ for short, is stable under the action of~$P$.
Moreover  the two matrices $P$ and $P^{\dag}$ share the same Perron vector.
%We simply denote $Q_{\pi,A}$,   $ \mu_{\pi} (A) $, and  $ \alpha_{\pi} (A) $ by $Q_{A}$, $ \mu (A) $, and $\alpha (A)$,
%	respectively.

The matrix $P$ is said to be $\pi$-{\em self-adjoint}   if $ P^{\dag_{\pi}} = P$.	
A simple argument based on the unicity of the Perron vector of an irreducible matrix shows that
	if $P$ is  $\pi$-self-adjoint, then $\pi$ is $P$'s Perron vector, i.e., $\pi = \pi_{_P }$.
In this case, the matrix $P$  is said to be \emph{reversible}.

\subsection{A formula {\em \`a la Green}}

We start with an equality that is a generalization of Green's formula.

\begin{prop}\label{pro:greengen}
Let $\pi$ be any positive probability vector in $\IR^n$, and let  $L$ be a  square matrix of size~$n$.
If $L$ is $\pi$-self-adjoint and  $\one \in \ker (L)$, then for all vector $x\in \IR^n$, it holds that
	$$ \langle x,L.x \rangle _{\pi} = - \frac{1}{2}  \sum_{i \in[n]} \sum_{ j \in [n]}  \pi_i \, L_{i,j}  \, (x_i - x_j)^2 \, .$$ 
\end{prop}

\begin{proof}
First we observe that
	$$\begin{array}{lcl}
	\sum_{ i\,,\,j } \pi_i \, L_{i \, j} \, (x_i - x_j)^2  & = &  \sum_{i \neq j} \pi_i \, L_{i \, j} \, (x_i - x_j)^2 \\ \\
		 & = & \sum_{i \neq j } \pi_i \, L_{i \, j} \, x_i^2 +  \sum_{i \neq j} \pi_i \, L_{i \, j} \,  x_j^2 
					 -2  \sum_{i \neq j} \pi_i \, L_{i \, j} \, x_i \, x_j .  
				\end{array}	 $$
Because of the assumptions on  $ L $, the first two terms are both equal to $ - \sum_{i \in [n] }  \pi_i \, L_{i \, i} \,  x_i^2 $
	and so 
	$$ \sum_{i,\, j } \pi_i \, L_{i \, j} \, (x_i - x_j)^2 = 
			-2 \left(  \sum_{i } \pi_i \, L_{i \, i} \,  x_i^2 +   \sum_{i \neq j} \pi_i \, L_{i \, j} \, x_i \, x_j  \right) . $$
Besides, we have
	$$ \langle x, Lx \rangle_{\pi} =  \sum_{i ,\, j }  \pi_i \, L_{i \, j} \, x_i \, x_j 
						     =	 \sum_{i }  \pi_i \, L_{i \, i} \,  x_i^2 +   \sum_{i \neq j} \pi_i \, L_{i \, j} \, x_i \, x_j  $$
	and the lemma follows.
\end{proof}

\subsection{Norms on $\Delta^{\bot_{\pi}}$}

As an immediate consequence of the above proposition, we obtain that if $P$ is a  reversible stochastic matrix,
	then  the quadratic form
	$$\Q_{P}(x) = \langle x, x-  P x \rangle_{\pi_{_P }}  $$ 
	 is non-negative and its restriction to $\Delta^{\bot_{P}}$  is positive definite. 	
Moreover, $P$ has $n$ real eigenvalues $\lambda_1(P) , \dots, \lambda_n(P)$ that satisfy 
	$$ -1  \leq  \lambda_n (P) \leq \dots \leq  \lambda_2 (P)  \leq \lambda_1 (P) = 1 . $$
The Perron-Frobenius theorem shows that if, in addition, $P$ has a positive diagonal entry, then the first and 
	the last inequalities are strict.
	
Besides, we obtain the classical minmax characterization of the eigenvalues of reversible stochastic matrices.

\begin{lem}\label{lem:gamma}
Let $P$ be any reversible stochastic matrix, and let $\pi$ be its Perron vector.
For any positive real number $\gamma$, the two following assertions are equivalent
\begin{enumerate}
\item $ \lambda_2 (P)  \leq 1-\gamma$;
\item $\forall x\in \Delta^{\bot_{P}}, \ \Q_{P} (x) \geq \gamma  \,  \lVert x  \lVert^2_{\pi} $.
\end{enumerate}
In other words, 
	$  \lambda_2 (P)  = 1 -  \inf_{x\in \Delta^{\bot_{P}}\setminus \{ \zero\}}  \frac{ \Q_{P} (x)}{  \lVert x  \lVert^2_{\pi} }$.
\end{lem}	

\begin{proof}
Let  $\{ \varepsilon_1, \dots, \varepsilon_n\}$ be an orthonormal basis for the inner product 
	$ \langle . , . \rangle_{\pi} $ such that  $\varepsilon_1 = \one$ and for each index $i\in [n]$,
	 $$P \varepsilon_i = \lambda_i (P) \, \varepsilon_i  . $$
	 
Let $z_1, \dots, z_n$ the  components of~$x$ in this basis, namely,
	$$ x = z_1  \varepsilon_1 +  \dots + z_n\varepsilon_n \, .$$
Hence, $$ \Q_{P} (x) =   \sum_{i \in [n]} \big(1-\lambda_i (P)\big) z_i^2   $$
	which shows the equivalence of the two assertions in the lemma.	
\end{proof}

Another corollary of Proposition~\ref{pro:greengen} is the following inequality between the two norms~$\lVert . \lVert_{\pi}$ 
	and~$\N$ on~$\Delta^{\bot_{\pi}} $, where $\pi$ is any positive probability vector.
	
\begin{cor}\label{cor:N>Euclid}
If $\pi$ is  a positive probability vector, then the
	Euclidean norm $ \lVert . \lVert_{\pi}$ is bounded above on $\Delta^{\bot_{\pi}}$ by the semi-norm $\N/\sqrt{2}$, i.e., 
	$$\forall x \in \Delta^{\bot_{\pi}}, \  \N(x) \geq \sqrt{2} \,  \lVert x  \lVert_{\pi} .  $$
\end{cor}

\begin{proof}
Let us consider the orthogonal projector $\pi . \one^{\rm{T}}$  on $\Delta$, where $\pi$ is $P$'s Perron vector.
Thus,  for any vector in $x\in\Delta^{\bot_{P}}$, we have
	$$  \lVert x  \lVert_{\pi}^2  = \langle x, x-  \pi . \one^{\rm{T}} . \, x \rangle_{\pi} .$$
Since $\pi . \one^{\rm{T}} $ is stochastic and reversible, Proposition~\ref{pro:greengen} gives
	\begin{equation}\label{eq:variance}
	  \lVert x  \lVert_{\pi}^2  =  \frac{1}{2}  \sum_{i \in[n]} \sum_{ j \in [n]}   (x_i - x_j)^2 \pi_i \, \pi_j 
	  \end{equation}
	and the inequality $ \N(x) \geq \sqrt{2} \, \lVert x  \lVert_{\pi} $ immediately follows. 
\end{proof}

\section{The spectral gap of  a reversible  stochastic matrix}\label{sec:eigen}

\subsection{An analytic bound}

We start by introducing the following  notation: given a stochastic matrix $P$ and its Perron  vector~$\pi$, we set 
	%\begin{equation}\label{eq:mu}
	$$ \mu (P) =  \min_{\emptyset \subsetneq S \subsetneq [n]}  \left (  \sum_{i \in S} \sum_{j\notin S} \pi_i \, P_{i \, j} \! \right )  .  $$
	%\end{equation}
		
\begin{lem}[Lemma 8 in~\cite{NOOT09}]\label{lem:in}
If $P$ is a  reversible stochastic matrix, then for every  vector $x \in \IR^n$, 
	$$ \Q_{P}(x) \geq \frac{\mu (P)}{n - 1} \,  \big( \N(x) \big)^2  .$$
\end{lem}

\begin{proof}
Using index permutation, we assume that  $x_1\leq \dots \leq  x_n$.
Since for any nonnegative numbers $v_1, \dots,v_k$, we have
	$$ (v_1 + \dots + v_k)^2 \geq v_1^2 + \dots + v_k^2 \, ,$$
	it follows that  
	$$ \sum_{i < j} \pi_i P_{i j} (x_i - x_j)^2  \geq   
	\sum_{i < j} \left( \pi_i P_{i j} \sum_{d=i}^{j-1} (x_{d+1} - x_d)^2 \right ) .
	$$
By reordering the terms in the last sum, we obtain
	$$ \sum_{i < j} \pi_i P_{i j} (x_i - x_j)^2  \geq   
	      \sum_{d=1}^{n-1} \, \sum_{i =1}^d 
	      \sum_{j= d+1}^n  \pi_i P_{i j} (x_{d+1} - x_d)^2 \,.$$
Then  Proposition~\ref{pro:greengen}  shows that 
	$$\Q_{P}(x) \geq  \mu (P) \,  \sum_{d=1}^{n-1} (x_{d+1} - x_d)^2 \, .$$
By Cauchy-Schwarz, we have
	$$ \sum_{d=1}^{n-1} (x_{d+1} - x_d)^2  \geq \frac{1}{n - 1 } \left(x_n - x_1 \right)^2 \, ,$$
	which completes the proof.
\end{proof}

That leads us to introduce
	\begin{equation}\label{eq:eta}
	\eta(P) = \frac{ n - 1 }{2 \, \mu  (P) }.
	\end{equation}
Combining Corollary~\ref{cor:N>Euclid} with Lemmas~\ref{lem:gamma} and~\ref{lem:in}, we obtain the following 
	lower bound on the spectral gap of a reversible stochastic matrix.

\begin{prop}\label{prop:specgapnoot}
If $P$ is a reversible  stochastic matrix, then
% The eigenvalues of a reversible  stochastic matrix $P$ other than 1 are bounded above by
	$$  \lambda_2(P) \leq  1 - \frac{1 }{ \eta(P) } $$
	with  $\eta (P)$ defined by~(\ref{eq:eta}).
\end{prop}

The quantity $\mu(P)$ is related to the \emph{Cheeger constant} 
	$$ h(P) = \min_{\pi(S) \leq 1/2} \frac{ \sum_{i \in S} \sum_{j\notin S} \pi_i \, P_{i \, j} }{\pi(S)} $$
	and satisfies $\mu(P)\leq h(P)/2$.
Cheeger's inequalities
	\begin{equation}\label{eq:cheeger}
	1 - 2 h(P) \leq \lambda_2(P) \leq 1 - \frac{h(P)^2}{2}
	\end{equation}
	give an estimate of the second eigenvalue of~$P$.
The bound $1-1/\eta(P)$ in Proposition~\ref{prop:specgapnoot} 
	 is incomparable with $1-h(P)^2/2$,
	but turns out to be worse in most cases\footnote{%
	If  $\mu(P) \geq 1/(n-1)$, then $1- h(P)^2/2 \leq 1- 1/\eta(P)$. 
	This inequality also holds in all the examples in Section~\ref{sec:applis}.}.
Moreover, computing $\mu(P)$, or equivalently  $\eta(P)$, is as difficult as computing $h(P)$ in general  
	-- so why presenting the bound $1-1/\eta(P)$?
In fact, our  primary motivation here is developed in Section~\ref{sec:singular}:
	 the latter bound gives a  simple  estimate on the singular
	values of even non-reversible stochastic matrices.
	
\subsection{A geometric bound}\label{sec:bcha}

Following~\cite{DS91}, we define the  $P$-\emph{length} of a path $\gamma  = u_1, \dots, u_{\ell +1}$ 
	in the graph~$G_{\! _P}$ by
	$$ \len P = \sum_{k\in[\ell]} \big( \pi_{u_k}  P_{u_k \, u_{k+1}} \big)^{-1} .$$

For our  geometric bound, we consider a family of paths in the graph~$G_{\! _P}$ defined as follows: 
	 for each pair of nodes~$i,j$, let~$\Gamma_{i , j}$ be a non empty set of 
	edge-disjoint paths from~$i$ to~$j$.
Since $P$ is irreducible, such a set exists. 
Moreover, Menger's theorem shows that~$\Gamma_{i , j}$ may be chosen with cardinality equal to any integer in $[\kappa]$, 
	where $\kappa$ is the edge-connectivity 
	of~$G_{\! _P}$\footnote{%
	The \emph{edge-connectivity} of a directed graph~$G$ is defined to be the minimum number 
	of edges in $G$ whose removal results in a directed graph that is not strongly connected.}.
As will become clear, the quality of our estimate depends on making a judicious choice for the path sets $\Gamma_{i , j}$.

The geometric quantity that appears in our bound is 
	\begin{equation}\label{eq:kbcb}
	\kbcb (P) = \max_{i \neq j }  \left( \sum_{\gamma \in \Gamma_{i , j} } \len P ^{\  -1}\right)^{-1} .
	\end{equation}

\begin{prop}\label{prop:specgapcha}
If $P$ is a reversible  stochastic matrix, then
	$$  \lambda_2(P) \leq 1 - \frac{1}{\kbcb (P) }$$ 
	where $\kbcb (P)$ is defined by (\ref{eq:kbcb}).
\end{prop}

\begin{proof}
Let $i$ and $j$ be any pair of distinct nodes.
Proposition~\ref{pro:greengen}  shows that 
	$$ \Q_{P}(x) \geq  \frac{1}{2} \sum_{\gamma \in \Gamma_{i , j} }  \sum_{(u , v) \in \gamma} \pi_u \, P_{u v } (x_u - x_v)^2  , $$
	where $\pi$ denotes the Perron vector of $P$.
By convexity of the square function,  we have
	$$  \left(\sum_{(u , v) \in \gamma}  (x_u - x_v)\right )^2 \leq  \sum_{(u , v) \in \gamma} \pi_u \, P_{u v } (x_u - x_v)^2 
			                                                              \sum_{(u , v) \in \gamma} \frac{1}{\pi_u \, P_{u v }} , $$
which implies 
	$$ \Q_{P}(x) \geq \left (\sum_{\gamma \in \Gamma_{i , j} } \frac{ 1}{ \len P} \right)   \frac{ (x_i - x_j)^2 }{2} 
	                  \geq  \frac{ (x_i - x_j)^2 }{2 \, \kbcb (P) } . $$

Hence 
	$$ \Q_{P}(x)  = \sum_{i \in[n]} \sum_{ j \in [n]}  \Q_{P}(x) \pi_i \pi_j\geq 
		\frac{1}{\kbcb (P)}  \left ( \frac{1}{2}   \sum_{i \in[n]} \sum_{ j \in [n]}   (x_i - x_j)^2 \pi_i \, \pi_j   \right) \!
		= \frac{1}{\kbcb (P)} \,  \lVert x \lVert_{\pi}^2 ,$$
		and the result follows.
The first equality holds because the sum of  $\pi$'s entries is 1 and the second one is the formula~(\ref{eq:variance}).
\end{proof}

%By definition of edge-connectivity, the out-boundary of every proper subset of nodes~$S$,
%	i.e., the set of edges $(i,j)$ with $i\in S$ and $j\notin S$, is of cardinality at least $\ka(G)$.
%It follows that 
%	$$\mu_{\pi}(M)\geq \ka(G_{\! _M}).\, \alpha_{\pi}(M) .$$
%Proposition~\ref{prop:specgapnoot} then leads to the lower bound $ 2 \alpha_{\pi} (M). \,\ka\,(G_{\! _M}) /(n -1) $ on the spectral gap of~$M$,
%	which is improved by Proposition~\ref{prop:specgapcha} in the case of a directed graph with normalized diameter less than 
%	$ \frac{n - 1}{2\, .\, \ka(G)}$.   , denoted $\ka(G)$,

Let us now recall some notions from graph theory~(see, e.g.,~\cite{HL94}).
First, define the \emph{depth} of a set of paths in a directed graph~$G$ as the maximum length of all its paths.
For every positive integer~$k$ and every pair of nodes~$(i,j)$, the $k$-\emph{distance from $i$ to $j$}, denoted~$d_k(i,j)$, 
	is  the minimum depth of the sets of pairwise disjoint-edge paths from $i$ to~$j$ of cardinality~$k$,
	if there is any;
	otherwise, the $k$-distance from $i$ to $j$ is infinite.
Then the $k$-\emph{diameter of~$G$}, denoted~$\diam_k(G)$, is the maximum $k$-distance between any pair of nodes.
The 1-diameter of~$G$  thus coincides with its diameter.

The parameter that naturally emerges when one looks for estimates of $\kbcb(P)$
	is the \emph{normalized diameter} of~$G$, denoted~$\diamnorm(G)$, defined by
	\begin{equation}\label{eq:diamnorm}
	 \diamnorm(G) = \min_{k\geq 1} \, \frac{ \diam_k(G)}{ k } .
	 \end{equation}
It clearly satisfies $  \diamnorm(G) \leq \diam(G)$.
Moreover, Menger's theorem shows that $\diam_k (G)$ is finite  if and only if~$k  $ is 
	less or equal to the edge-connectivity of~$G$, denoted~$\ka(G)$, thus providing the upper bound
	$  \diamnorm(G) \leq (n-1)/ \ka(G)$.
	
Let~$k$ be any integer such that $1 \leq k \leq 	\ka(G)$.
For every set $\Gamma_{i , j}$ of $k$ edge-disjoint paths from~$i$ to~$j$, we have
	$$ \sum_{\gamma \in \Gamma_{i , j} } \frac{1}{\len {} }  \geq  
				\frac{k}{d_k(i,j)} \geq  \frac{k}{  \diam_k( G_{\! _P} ) } . $$
It follows that if $k$ realizes the minimum in~(\ref{eq:diamnorm}), then 
	$$ \sum_{\gamma \in \Gamma_{i , j} } \frac{1}{\len {}}  \geq   \frac{1}{\diamnorm  (G_{\! _P} )} .$$
By setting
	\begin{equation}\label{eq:alpha}
	\alpha  (P) = \min_{(i,j)\in E(G_{_P })}    \pi_i P_{i j}   ,
	\end{equation}
		we have 
	$ \len P \leq \len {}   / \alpha(P) $, and hence
	$$ \kbcb(P) \leq    \frac{\diamnorm  G_{\! _P} )}{\alpha(P)}  .$$
Thus, we obtain the following corollary to Proposition~\ref{prop:specgapcha}.
	
\begin{cor}\label{cor:Kbcb}
The eigenvalues of a reversible stochastic matrix smaller than 1 are bounded above~by
	$$  \bcb (P) = 1 - \frac{\alpha(P)}{\diamnorm (G_{\! _P}) } ,$$ 
	where $\alpha(P)$ is defined by (\ref{eq:alpha}) and $\diamnorm (G_{\! _P}) $ is the normalized diameter 
	of the graph associated to $P$.
\end{cor}	
\subsection{Diaconis and Stroock's geometric bound}

We now present another geometric bound on the spectral gap of a reversible stochastic matrix,
	which  has been 
	developed by  Diaconis and Stroock~\cite{DS91}. 
It depends on the choice of a set of paths in the directed graph $G_{\! _P} $,
	one for each ordered pair of distinct nodes:
	for every pair $i,j$ of nodes, let~$\gamma_{i\, j}$ be  a path from~$i$ to~$j$,
	and let~$\Gamma$ be the set of all these paths.

The geometric quantity that appears in their bound is
	\begin{equation}\label{eq:kds}
	\kDS (P) = \max_{ e }  \sum_{ e \in \gamma_{i \, j} }  \lenij P  \pi_i \,  \pi_j  ,
	\end{equation}
 	where the maximum is over  edges in the directed graph~$G_{\! _P}$ and the sum is over all paths in~$\Gamma$
	that traverse~$e$.

 Diaconis and Stroock~\cite{DS91}  developed a discrete analog of the Poincar\'e's inequality for
	estimating the spectral gap of the Laplacian on a domain:

\begin{prop}[Proposition 1 in~\cite{DS91}]\label{prop:specgapDS}
If $P$ is a reversible stochastic matrix, then 
	$$ \lambda_2(P) \leq   1 - \frac{1}{\kDS (P)}$$ 
	where $\kDS (P) $ is defined by (\ref{eq:kds}).
\end{prop}	

As for our bound which depends on the choice of the path sets $\Gamma_{i j}$, the quality of their estimate 
	depends on the choice for the paths $\gamma_{i , j}$: the lower bound~$\kDS (P) $ is all the better 
	if selected paths do not traverse any one edge too often. 
Following~\cite{DS91}, every path~$\gamma_{i , j}$ is chosen to be a geodesic.
The geometric quantity that  arises  here is a measure of bottlenecks in~$G_{\! _P}$ defined as
\begin{equation}\label{eq:bdef}
	b(G_{\! _P}) = \min_{\Gamma} \, \max_{e } \left | \{ \gamma \in \Gamma : e \in \gamma \} \right |  ,
\end{equation}
	where  the minimum is over the sets of paths~$\Gamma$ containing  only geodesics, and 
	 the maximum is over all the edges of $G_{\! _P}$.
It can be shown that 
	\begin{equation}\label{eq:bottleneck}
	 \frac{n-1}{\ka(G_{\! _P})} \leq b(G_{\! _P}) \leq n^2 ,
	 \end{equation}
%	and bounds are both tight.
	where $\ka(G_{\! _P}) $ is the edge-connectivity of $G_P$.
(The second inequality is straightforward;
	the first one may be proved by considering the partitioning of $G_P$ into two 
	strongly connected components when removing a certain set of ${\ka(G_{\! _P})}$ edges.)
Hence, $ \diamnorm(G_{\! _P}) \leq b(G_{\! _P})$.

Like the first geometric bound~$1-1/\kbcb(P)$,   the bound $1-1/\kDS(P)$ can be usefully approximated as follows.	
\begin{cor}\label{cor:KDS}
The eigenvalues of a reversible stochastic matrix~$P$ other than 1 are upper-bounded~by
	$$  \DS (P)  =  1 - \frac{ \alpha (P) }{ ( \pi_{\max})^2  \, \diam(G_{\! _P})  \, b (G_{\! _P}) } $$ 
	where $\alpha(P)$ is defined by (\ref{eq:alpha}), 
	 $\pi_{\max}$ is the largest entry of the Perron vector of~$P$,
	 $\diam(G_{\! _P}) $ and $b(G_{\! _P})$ are the diameter and the bottleneck measure 
	of the graph associated to $P$, respectively.
\end{cor}	

%%%%%%%%%%%%%%%%%%%%%%%%%%%
%Both bounds are weak for directed graphs with  ``bottlenecks'' like trees.
%They compare well with our geometric bound $\bcb$ when $\pi_{\max}$ is equal to its minimum value,
%	namely~$1/n$.
%%%%%%%%%%%%%%%%%%%%%%%%%%%%%%	
\section{Upper bounds on the second singular value of a stochastic matrix}\label{sec:singular}

%Then we examine the case where the matrix is of the form~$M = A^{\dag_{\pi}}A$ with $A$ a stochastic matrix,
%	and improve  the graph-based bound  in the \emph{reversible} case, i.e., when$A$ is itself $\pi$-self-adjoint.
	
Let $A$ be any irreducible	stochastic matrix of size $n$ with positive diagonal entries.
If $\pi$ is the Perron vector of $A$, then the matrix $A^{\dag} A$ is  also stochastic,
	and the three stochastic matrices $A$, $A^{\dag}$, and $A^{\dag} A$
	share the same Perron vector~$\pi$.
Moreover, $A^{\dag} A$ is reversible and has $n$ non negative eigenvalues.

Propositions~\ref{prop:specgapnoot},~\ref{prop:specgapcha}, and~\ref{prop:specgapDS}
	provide lower bounds on the spectral gap of $ A^{\dag}A$, 
	 which involve  the positive coefficients % $ \mu(A^{\dag}A)$ and 
	$ \pi_i (A^{\dag}A)_{i j}$ when positive.
These coefficients are  roughly  bounded below by~$ \alpha(A)^2/\pi_{\max}$ with 
	$\pi_{\max} = \max_{i \in [n]} \pi_i$ and  $\alpha(A)$ defined by (\ref{eq:alpha}).
	
Interestingly, a generalization of a result in~\cite{NOOT09} combined with Proposition~\ref{prop:specgapnoot}
	gives an analytic bound on the spectral gap that is linear in the coefficient~$ \alpha(A)$ 
	and that holds even when $A$ is non reversible. 
In the case the matrix $A$ is reversible,  a lower bound on the spectral gap of~$A$ easily provides 
	a lower bound on the spectral gap of~$A^{\dag}A$.

\subsection{Analytic bound}

We start with a lemma that has been established in~\cite{NOOT09} under the more restrictive assumption of doubly stochastic
	matrices.
%For the convenience of the reader, we provide some details of the proof.

%\begin{lem}[Lemma 5 in~\cite{NOOT09}]\label{lem:cut}
\begin{lem}\label{lem:cut}
If $A$ is an irreducible stochastic matrix, then $$ \mu(  A^{\dag}  A ) \geq \alpha (A)/2 \, .$$
\end{lem}

\begin{proof}
Let $S$ be  any non empty subset of $ [n] $.
Since $A$ is a stochastic matrix, for every index $k \in [n]$, either $\sum_{i \in S} A_{k i } > 1/2$ or
	$\sum_{j\notin S} A_{k j} \geq 1/2$, and the two cases are exclusive, that is,
	the two subsets of $[n]$ defined by 
	$$ S^+ = \{ k \in [n] : \sum_{i \in S} A_{k i } > 1/2 \} \ \mbox{ and } \   S^- = \{ k \in [n] : \sum_{j \notin S} A_{k j } > 1/2 \} $$
	satisfy $ S^- =  [n] \setminus  S^+$.
Hence,  %\begin{equation}\label{eq:NOOT}
		$$ \sum_{i \in S} \sum_{j\notin S} \pi_i \, \big (A^{\dag}  A \big)_{i \, j} 
		= \sum_{k \in [n]}  \sum_{i \in S} \sum_{j\notin S}  \pi_k A_{k i} A_{k j} 
		\geq \frac{1}{2}
		\left (  \sum_{k \in S^+ }  \sum_{j\notin S}  \pi_k  A_{k j}  + 
		          \sum_{k \in S^- }  \sum_{i \in S} \pi_k A_{k i}              \right) . $$
	       %	\end{equation}
Then we consider the two following cases:
\begin{enumerate}
\item Either $ S^- \cap S \neq \emptyset$ or $ S^+ \cap ([n] \setminus S )\neq \emptyset $. 
If $\ell$ is in one of these two sets, then we obtain that
	$$  \sum_{i \in S} \sum_{j\notin S} \pi_i \, \big (A^{\dag}  A \big)_{i \, j} \geq \frac{ \pi_{\ell}  A_{\ell \ell}  }{2} .$$
		 
\item Otherwise,  $S^+ = S$.
Since $A$ is irreducible, the non-empty set $S$ has an outgoing edge $(k_1, j)$  and an  incoming edge  $(k_2,i)$
	in~$G_{_A}$.
It follows  that
	$$ \sum_{i \in S} \sum_{j\notin S} \pi_i \, \big (A^{\dag}  A \big)_{i \, j} \geq \frac{1}{2}
		\left (   \pi_{k_1}  A_{k_1 j}  +  \pi_{k_2} A_{k_2 i}              \right) . $$ 
\end{enumerate}
In both cases, we arrive at $ \sum_{i \in S} \sum_{j\notin S} \pi_i \, \big (A^{\dag}  A \big)_{i  j}  \geq \alpha(A)/2$.
\end{proof}

Applied to the stochastic matrix~$A^{\dag} A$, Proposition~\ref{prop:specgapnoot} takes 
	the form:
\begin{prop}\label{prop:specgapgen}
Let $A$ be an irreducible stochastic matrix with a positive diagonal.
The  matrix $A^{\dag} A$ has $n$ real eigenvalues that satisfy
	$$ 0 \leq \lambda_n(A^{\dag} A ) \leq \dots \leq  \lambda_2(A^{\dag} A ) \leq  
	 1-\frac{\alpha(A)}{ n-1 }    <  \lambda_1(A^{\dag} A )  = 1.$$
\end{prop}

\subsection{The reversible case}

If the  stochastic matrix $A$ with positive diagonal is reversible, then 
	the $n$ eigenvalues of $A$ are all real and the Perron-Frobenius 
	theorem implies that  % they  satisfy	
	$$ -1  <  \lambda_n (A) \leq \dots \leq \lambda_2(A) <  \lambda_1(A) = 1 .$$

Similarly, the stochastic matrix $A^{\dag} A = A^2$ has $n$ real eigenvalues which, written in decreasing order, satisfy	
	$$ 0  \leq  \lambda_n (A^{\dag} A ) \leq \dots \leq \lambda_2(A^{\dag} A ) <   \lambda_1(A^{\dag} A ) = 1 . $$
Hence $ \lambda_2(A^{\dag} A ) =  \max ( \vert \lambda_n(A) \vert^2,  \vert \lambda_2 (A)\vert^2 ) $.

Propositions~\ref{prop:specgapnoot},~\ref{prop:specgapcha}, and~\ref{prop:specgapDS} show that 
	$$    \lambda_2(A) \leq 1 - \frac{ 1 }{ \min ( \eta(A), \kbcb(A), \kDS(A) )} .$$
Computing $\eta(A)$ % (as well as the Cheeger constant $h(A)$ , see e.g.,~\cite{Lub89}) i
	is difficult in general %  and often yield 
	and thus we keep on just with the two geometric bounds $\kbcb(A)$ and $ \kDS(A) $.

Every eigenvalue of~$A$ lies within at least one Gershgorin disc $D(A_{i i}, 1- A_{i i})$, 
	and thus
	\begin{equation}\label{eq:gershgorin}
	   -1 + 2 \, a (A)  \leq  	\lambda_n(A)  
	 \end{equation}
	where $a (A) = \min_{i \in [n]}   A_{i \, i}  $.
	
\begin{prop}\label{prop:specgapreversible}
Let $A$ be a irreducible stochastic matrix with a positive diagonal.
If $A$ is reversible, then the stochastic matrix $A^{\dag} A$ has $n$ real eigenvalues that satisfy
	$$ 0 \leq \lambda_n(A^{\dag} A ) \leq \dots \leq  \lambda_2(A^{\dag} A ) \leq  
	\left( 1 -  \min \left ( 2 \, a (A) \, , \,  \frac{1}{ \min ( \kbcb(A) , \kDS(A)  ) }  \right ) \right)^2 <  \lambda_1(A^{\dag} A )  = 1 $$
	 with $a(A) = \min_{i \in [n]}   A_{i \, i}  $, 
	 $\kbcb(A)$, and $\kDS(A)$ defined by  (\ref{eq:kbcb}) and (\ref{eq:kds}), respectively.
\end{prop}	

Every path $\gamma$ in $G_A$ satisfies 
	$$ |\gamma |_{A} \leq \frac{ | \gamma |}{\alpha(A)} $$
	where $| \gamma | $ denotes $\gamma$'s length and $\alpha(A)$ is defined by~(\ref{eq:alpha}).
If each path set $\Gamma_{i j}$ is reduced to one shortest path from $i$ to $j$, then 
	$$ \kbcb(A)  \leq \frac{n-1}{\alpha(A)} .$$
Proposition~\ref{prop:specgapreversible}  thus  improves the 
	general bound of $ \alpha(A)/( n - 1 )$ in Proposition~\ref{prop:specgapgen} 
	when $A$ is reversible.	
%In the reversible case, the two best lower bounds on the spectral gap that we have obtained are 
%	 the analytic bound of $  \mu(A^{\dag}A)/n$
%	 and the graph-based bound in Proposition~\ref{prop:specgapreversible}.
%	 
%The above arguments applied to the bound in Proposition~\ref{prop:specgapnoot} 
%	lead to the bound $ \max \left\{ 1- 2 \, \diamnorm (A) ,  1- \mu(A)/n \right \} $, and with Lemma~\ref{lem:mualpha}
%	to $ \max \left\{ 1- 2 \, \diamnorm (A) ,  1- c(A).\alpha(A)/n \right \} $, which 
%	is of the same order of magnetude as the bound greater or equal 
%	is not better than 
%	the graph-based bound of $ \max \left\{ 1- 2 \, \delta (A) ,  1- \frac{c(A).\alpha(A)}{d(A) }\right \} $
%	in the general case since we may only claim that $d(A)\leq n$ and $\mu(A) \geq c(A). \alpha(A)$
%	without any additional information on the matrix~$A$.
	
\section{Averaging algorithms and convergence rates}

\subsection{Averaging algorithms, stochastic matrices and asymptotic consensus}

We consider a  discrete time system of $n$ autonomous agents, denoted $1, \dots,n$, 
	connected via a network that may change over time.
Communications at time $t$ are modelled by a directed graph~$\dG(t)= ([n],E(t))$.
Since an agent can communicate with itself instantaneously, there is  a self-loop at each node
	in every graph~$\dG(t)$.
The sets of incoming and outgoing neighbors of the agent~$i$ in $\dG(t)$ are denoted by~$\In_i(t)$ 
	and~$\Out_i(t)$, respectively.	
The sequence $\dG = \big( \dG(t)\big)_{t\geq 1}$ is called  \emph{the dynamic communication graph},
	or just the \emph{communication graph}.

In an averaging algorithm ${\cal A}$, each agent~$i$ maintains a local variable $x_i$, initialized to some scalar value $x_i(0)$,
	and applies an update rule of the form
	\begin{equation}\label{eq:calA}
	x_i(t) = \sum_{k \in \In_i(\dG(t))} A_{i k}(t) \, x_k(t-1) 
	\end{equation}
	with $A_{i k}(t)$ which are all positive and $\sum_{k \in In_i(\dG(t))} A_{i k}(t) = 1$.
The algorithm ${\cal A}$ precisely consists in the choice of the weights~$A_{i k}(t)$; typical  averaging algorithms
	are examined in Section~\ref{sec:applis}.
The update rule (\ref{eq:calA}) corresponds to the 	equation
	$$ x(t) = A(t) \, x(t-1)$$
	where $A(t)$ is the $n \times n$  stochastic matrix whose $(i,k)$-entry is the weight $A_{i k}(t)$ if $(k,i)$ is an edge in $\dG(t)$,
	and 0 otherwise. 
Hence, the directed graph associated to the matrix~$A(t)$ is the reverse graph of~$\dG(t)$.
			
An execution of ${\cal A}$ is totally determined by the initial state $x(0) \in \IR^n$
	 and the   communication graph $\dG$. 
We say that ${\cal A}$ {\em achieves asymptotic consensus} in an execution if the sequence~$x(t)$ converges
	to a vector $x^*$ that is colinear to $\one = (1,\dots,1)^{\rm{T}}$.
The {\em convergence rate} in this execution is  defined as 
	$$  \varrho =  \lim_{t\rightarrow \infty} \lVert x(t) - x^* \lVert^{1/t} $$
	where $ \lVert .  \lVert $ is any norm on $\IR^n$.
%We also refer to the {\em convergence time}, $T(\varepsilon)$, defined by
%	$$T(\varepsilon) = \inf \left( \tau :  \forall t \geq \tau,\  \forall x(0) \in \IR^n , \ \N(x(t)) \leq \varepsilon \N(x(0)) \right ). $$
	
The classes of averaging algorithms under consideration and their executions
	are restricted by the following assumptions.
\begin{description}
\item[]A1:  All the directed graphs $\dG(t)$ have a self-loop at each node and are strongly connected.
\item[]A2:  There exists some positive lower bound on the positive entries of the matrices  $A(t)$.
%\item[A3:]  Every stochastic matrix  $A(t)$ is reversible. \item[A4:] 
\end{description}
Observe that A1 is equivalent to the fact that every matrix~$A(t)$ has a positive diagonal
	and is ergodic. 
As an immediate consequence of the fundamental convergence results in~\cite{Mor05,CMA08a}, 
	we have that asymptotic consensus is achieved in every run of an averaging algorithm satisfying~A1-2. 

\subsection{Case of a constant Perron vector}

Our first results concern  executions  that satisfy the following assumption in addition to~A1-2.
\begin{description}
\item[]A3:   All the matrices  $A(t)$ share the same Perron vector $\pi$.
\end{description}
Observe that under the assumption A3, the limit vector $x^*$, if exists, is equal to 
	$\sum_{i\in [n]} \pi_i x_i(0) \one$.
	
The assumption  A3 holds for time-varying communication graphs 
	that arise in diverse classical averaging algorithms  (e.g., see Section~\ref{sec:applis}).
Besides,  the validity of A2 and A3 allows us to introduce the two positive infima
	\begin{equation}\label{eq:aandalpha}
	a = \inf_{i \in [n] }  A_{i \, i} (t)  \ \mbox{ and } \ 
	 \alpha = \inf_{(i,j) \in E(t) } \pi_i A_{i \, j}(t) .
	 \end{equation}
The inequality  (\ref{eq:gershgorin}) shows that all the  eigenvalues of  the matrices $A(t)$  are uniformly 
	bounded below by $ - 1 + 2 a > -1$.
Moreover, since the number~$n$ of agents is fixed, the quantities  $\eta(A(t))$,
	 $\kbcb(A(t))$, and $\kDS(A(t))$ defined by (\ref{eq:eta}), (\ref{eq:kbcb}), and~(\ref{eq:kds}),
	 respectively, 
	 are uniformly  bounded from the above.
	 
\begin{thm}\label{thm:convergencerate}
In any of its executions  satisfying the  assumptions A1-3, 	 an averaging 
	algorithm achieves asymptotic consensus with a convergence rate
	$$ \varrho \leq \sup_{t\geq 1} \, \sqrt{ \lambda_2 \big ( A(t)^{\dag}A(t)  \big )} .     $$
	 \end{thm}

\begin{proof}
Let $y(t)$ denote the $\pi$-orthogonal of $x(t)$ on $\Delta^{\bot_{\pi}}$.
Since $A(t)^{\dag}$ is stochastic, then
	$$ \langle x(t),\one \rangle_{\pi}  = \langle A(t) \, x(t-1),  \one \rangle_{\pi} =  \langle x(t-1),  \one \rangle_{\pi}. $$
Therefore, the orthogonal  projection of $x(t)$ on $\Delta$ is constant and $y(t) = A(t) y(t-1)$.
Let $\V(t)$ be the variance of~$x(t)$, that is
	 $$\V(t) = \lVert x(t) - a \one \lVert^2_{\pi} = \lVert  y(t)  \lVert^2_{\pi}\,  $$
	with $a =       \langle x(0),\one \rangle_{\pi} $.                         
Then 
	$$  \V(t-1) - \V(t)  =      \left \langle y(t - 1), y(t-1) \right \rangle_{\pi} - \left  \langle A(t)  y (t-1 ),A(t)   y (t -1 )  \right  \rangle_{\pi}     
%	                      =     \langle y(t - 1), y(t-1)- A(t)^{\dag} A(t) . y (t-1 )  \rangle_{\pi}   $$  
	 			= \Q_{A(t)^{\dag}A(t)} \left( y(t-1) \right) .$$
				                                            
By Proposition~\ref{pro:greengen}, it follows that $\V$ is non-increasing.
Moreover, the variational characterization in Lemma~\ref{lem:gamma} shows that
	$$ \V(t) \leq \beta^t \, \V(0) \, ,$$	
	where $\beta$ is any uniform upper bound on the second largest eigenvalues of the matrices $ A(t)^{\dag}A(t)$.
\end{proof}

In addition to A1-3, we may assume permanent reversibility.
\begin{description}
\item[] A4:  All the matrices  $A(t)$ are reversible.
\end{description}

\begin{cor}\label{cor:reversiblerate}
In any of its executions  satisfying the  assumptions A1-4, 	 an averaging 
	algorithm achieves asymptotic consensus with a convergence rate
	$$ \varrho \leq 1 -  \min \left ( 2 \, a  \, , \,  \frac{1}{ \min (  \kbcb , \kDS ) }  \right )    , $$
	 where $a$ is defined by~(\ref{eq:aandalpha}), and 
	 $\kbcb$ and $\kDS$ are uniform upper bounds on $\kbcb(A(t))$ and  $\kDS(A(t))$.
\end{cor}

\begin{proof}
Let $\kbcb$ and $\kDS$ be uniform upper bounds on $\kbcb(A(t))$ and  $\kDS(A(t))$, respectively (assumptions A2-3).
Proposition~\ref{prop:specgapreversible} shows that for any positive integer $t$,
	$$ \lambda_2 \big ( A(t)^{\dag}A(t) \big ) \leq  \left (1 -  \min \left ( 2 \, a  \, , \,  \frac{1}{ \min (  \kbcb , \kDS ) }  \right ) \right)^2 . $$
The result immediately follows from Theorem~\ref{thm:convergencerate}.
\end{proof}

%We easily check that 	
%	$$\forall x \in \Delta^{\bot_{\pi}}, \  N(x) \leq 
%	    \frac{2 \, \lVert x  \lVert_{\pi}}{\sqrt{\pi_{\min}}}   $$ 
%	 where $\pi_{\min} = \min_{i\in[n]}(\pi_i)$.
%By using $\log(1-a) \leq -a$ when $a$ is any positive number, Theorem~\ref{thm:convergencerate} then implies	 that
%	\begin{equation}\label{eq:convergencetime}
%	T(\varepsilon) \leq \frac{2 \, \log \big(2/ (\varepsilon \sqrt{\pi_{\min} }) \big)}
%	                                            {\min\{2 a,\alpha/ \diamnorm \}} .
%	 \end{equation}
%	
If at every time~$t$, the matrix $A(t)$ is symmetric and  $\dG(t)$  is the complete graph, 
	then  Corollary~\ref{cor:reversiblerate} gives the bound 
	$$ \varrho \leq 1 -  \frac{\inf_{i , j \in [n]^2, \ t \geq 1 } A_{i j }(t)  }{ n }    .$$
This is the bound  developed by Cucker and Smale~\cite{CS07} to analyze the formation of flocks 
	in a population of autonomous agents which move together.
	
\subsection{Small variations of the Perron vector}

Theorem~\ref{thm:convergencerate} shows that in any execution 
	of the \emph{EqualNeighbor} algorithm -- where the weights and the entries of Perron vectors
	are bounded below by $1/n$  and  $1/n^2$, respectively
	(cf. Section~\ref{sec:applis}) -- 
	the convergence rate is in $1- O(n^{-3})$ if the Perron vector is constant.
With time-varying Perron vectors, no polynomial bound holds.
Indeed, Olshevsky and Tsitsiklis~\cite{OT11} proved that the convergence time of this averaging  algorithm 
	is exponentially large in an execution where the support of the communication graph
	is fixed but agents move from one node to another node: 
	in the $n/2$-periodic   communication graph formed with bidirectional 2-stars of size~$n$, 
	the convergence rate is larger than $1- 2^{3 - n/2}$ while entries of each Perron vector
	is greater than $1/6$ for the two centers and greater than  $2/3n $ for the other agents.
	
Our next result, which consists in an extension of Theorem~\ref{thm:convergencerate} to the case of a time-varying Perron vector,
	provides a heuristic analysis of convergence rates:
	an exponential convergence time as in the above example may occur 
	only if the Perron vector of the matrices~$A(t)$
	significantly vary over time.
		
\vspace{0.3cm}	
	
We start by weakening the  assumption A3.
\begin{description}
\item[]A3b:  Entries of the Perron vectors are uniformly lower bounded by some 
	positive real number.
\end{description}
Under the assumption A3b, the infima $a$ and $\alpha$ defined by~(\ref{eq:aandalpha}) are still
	positive.
Moreover, the quantity
	\begin{equation}\label{eq:beta}
	 \nu = \sup_{i \in [n] , \, t >0 }  \sqrt{\frac{\pi_i (t+1) }{\pi_i(t)}  }
	 \end{equation}
	 is finite.

\begin{thm}\label{thm:smallvariations}
In any of its executions  satisfying the  assumptions A1-2 and A3b, 
	 an averaging algorithm achieves asymptotic consensus with a convergence rate
	 $$ \varrho \leq \nu \,  \sup_{t\geq 1} \, \sqrt{ \lambda_2 \big ( A(t)^{\dag}A(t)  \big )} .     $$
%	 $$ \varrho \leq \nu  \left ( 1- \min \left ( 2 \, a  \, , \,  \frac{1}{ \min (  \kbcb , \kDS ) }  \right )     \right )   $$
%	 with $a$ and $\nu$  defined by~(\ref{eq:aandalpha}) and~(\ref{eq:beta}), respectively, and 
%	 where $\kbcb$ and $\kDS$ are uniform upper bounds on $\kbcb(A(t))$ and  $\kDS(A(t))$.
\end{thm}

\begin{proof}
For any norm $\Vert . \Vert$ on $\IR^n$, let $\Vert . \Vert_{\IR^n/\Delta}$ denote the quotient norm  on the quotient vector
	space~$\IR^n/\Delta$, given by
	$$ \Vert \left [x\right ] \Vert_{\IR^n/\Delta} = \inf_{y\in \left [x\right ]} \Vert y  \Vert $$
	where $  \left [x\right ] = x + \Delta$.
It will be simply denoted $ \Vert \left [x\right ] \Vert $, as no confusion can arise.
In the case of the Euclidean norm $\Vert . \Vert_{\pi}$, we have
	 $$ \Vert \left [x\right ] \Vert_{\pi} = \Vert y  \Vert_{\pi} , $$
	where $y$ is the orthogonal projection of $x$ onto  $\Delta^{\bot_{\pi}}$.
		
If $\Delta$ is  an invariant subspace of the linear operator $A: \IR^n  \rightarrow \IR^n$, then let 
	$\left [A\right ] :  \IR^n/\Delta \rightarrow \IR^n/\Delta$  denote the  corresponding quotient operator.
The operator norm of $ \left [A\right ] $ associated to quotient norm $\Vert . \Vert_{\pi}$ is defined 
	as~$ \Vert \left [A\right ]  \Vert_{\pi} = \sup_{\left [x\right ] \neq 0}  \left( \Vert \left [A\right ]  \left [x\right ] \Vert_{\pi} / \Vert \left [x\right ] \Vert_{\pi} \right)$.
One can easily check that 
	$$ \Vert \left [A\right ]  \Vert_{\pi} = 
	      \sup_{y \in \Delta^{\bot_{\pi}} \setminus\{0\}}  \frac{ \Vert  A \, y \Vert_{\pi} }{\Vert y \Vert_{\pi} }  ,$$
	     i.e.,  $ \Vert \left [A\right ]  \Vert_{\pi} = \Vert A/  \Delta^{\bot_{\pi}}  \Vert_{\pi} $.
Hence $ \Vert \left [A\right ]  \Vert_{\pi} = \sqrt{ \lambda_2(A^{\dag}A)}$.

Let $x\in \IR^n$, and let $\pi$ and $\pi'$ be two positive probability vector.
We easily get  that 
	$$ \Vert x  \Vert_{\pi'}^2  \leq  \Vert x  \Vert_{\pi}^2  \ \max_{i \in [n]} \  \frac{\pi_i'}{\pi_i}  ,$$
	which implies that
	$$ \Vert \left [x\right ]  \Vert_{\pi'}^2  \leq \Vert \left [x\right ]  \Vert_{\pi}^2  \ \max_{i \in [n]} \  \frac{\pi_i'}{\pi_i} .$$

Let us now introduce the quotient form of $\V(t)$ defined as 
	$$\W(t) = \Vert \left [ x (t) \right] \Vert_{\pi (t)}^2 . $$
Then we have 
	\begin{equation}\label{eq:recbeta}
	 \W(t) =  \Vert \left [A(t)\right ]  \, \left [x (t-1)\right ]  \Vert_{\pi (t)}^2 \leq 
	                                    \Vert \left [A(t)\right ] \Vert_{\pi (t)}^2 \,  \Vert \left [x(t-1)\right ] \Vert_{\pi (t)}^2 
	 \end{equation}
	 and thus
	 $$ \W(t)   \leq  \Vert  \left [A(t)\right ] \Vert_{\pi (t)}^2 \,   \max_{i \in [n]}  \ \frac{\pi_i (t)}{\pi_i (t-1)}  \ 
	                                    \W(t-1) , $$
	which completes the proof.
	\end{proof}

The bound  in Theorem~\ref{thm:smallvariations} allows for a qualitative 
	analysis of convergence time, but is quantitatively  trivial in most cases.
However, the recurring inequality~(\ref{eq:recbeta}) on which it is based may be also helpful
	from a quantitative viewpoint, e.g., for controlling convergence times in case the Perron vector eventually stabilizes~\cite{CBLM20}.

%\subsection{Strong connectivity with bounded delays}
%
%Strong connectivity may be guaranteed not at each time step, but over bounded periods of time.
%That corresponds to the existence of some positive integer~$B$ such that the cumulative communication graph
%	over the time interval~$[t,t +B-1]$, that is to say the directed 
%	graph~$\dG(t:t+B-1) = \dG(t) \circ \dots \circ \dG(t +B-1)$, is strongly connected.
%Thus the assumption A1 can be relaxed to
%\begin{description}
%\item[]A1b:  For every time~$t$,  the directed graph $\dG(t)$ has a self-loop at each node and 
%	$\dG(t:t+B-1) $ is strongly connected.
%\end{description}

\section{Metropolis, EqualNeighbor, and FixedWeight algorithms}\label{sec:applis}

We now examine three fundamental averaging algorithms, classically called  \emph{Metropolis}, \emph{EqualNeighbor},
	and  \emph{FixedWeight}, which all achieve asymptotic consensus if the (time-varying) topology is 
	permanently strongly connected.
While the EqualNeighbor algorithm is directly implementable in a distributed setting, FixedWeight  requires the agents to have 
	knowledge over time: the topology may vary, but each agent is supposed to know an upper bound on its in-degrees.
As for Metropolis, it requires the agents to have knowledge at distance one:  each agent is supposed to know 
	the current in-degree of its neighbors.
	
For each of these algorithms, the Perron vectors are constant for large classes of   time-varying topologies:
	when the communication graph is permanently bidirectional this holds for the Metropolis algorithm, 
	when it is permanently Eulerian,  %footnote{ (i.e., strongly connected and each node has an in-degree  equal to its out-degree) 
	for the FixedWeight algorithm, and when it is permanently Eulerian with constant (in time or in space) degrees, for  EqualNeighbor.
In each of these cases, the corresponding stochastic matrices are all reversible and  thus 
	 Corollary~\ref{cor:reversiblerate} applies. 	
	
 \subsection{Algorithms and simplified bounds}
 
First, let us fix some notation.
If $p(G)$ denotes any parameter of a directed graph~$G$, let $p(\dG)$ denote the associated  parameter 
 	for the dynamic graph~$\dG$ defined as
	$$ p(\dG) = \sup_{t\geq 1} \, p(\dG(t)) .$$
For instance,  $\diam(\dG)$ denotes  the diameter of $\dG$, $\diamnorm(\dG)$ its normalized diameter, and
	$\bottle(\dG)$ its bottleneck measure.
Similarly, if $d_i(t)$ denotes the  in-degree of~$i$ in~$\dG(t)$ and~$\dmax(t)$ the maximum in-degree in this graph
	(i.e., $\dmax(t) = \max_{i \in [n]}d_i(t)$), then 
	$$ \dmax(\dG) = \max_{i \in [n], \ t\geq 1} \, d_i(t) .$$

\paragraph{Metropolis algorithm with a time-varying  bidirectional topology.}
Weights in the  the Metropolis algorithm are given by
	$$ M_{i j}(t) = \left\{ \begin{array}{ll}
	                     \frac{1}{\max \left ( d_i (t), d_j(t)\right)} & \ \mbox{ if } \ j\in  \In_i (t) \setminus \{i\} \\
	                     1- \sum_{ j\in {\cal N}_i (t) \setminus \{i\} } \frac{1}{\max \left ( d_i (t), d_j(t)\right)}    & \ \mbox{ if } \ j= i  \\
	                     0 & \ \mbox{ otherwise .} 
	                     \end{array}  \right.$$
%	 where ${\cal N}_i(t)$ denotes the set of $i$'s in-neighbors in the communication graph $\dG(t)$, 
%	 and $d_i(t) = | \, {\cal N}_i(t) |$.

If  $\dG(t) $ is bidirectional,  then the matrix $M(t)$ is symmetric, and so doubly stochastic.
Its Perron vector  is $(\frac{1}{n}, \dots, \frac{1}{n})^{\rm{T}}$.
In any execution of Metropolis with a   communication graph that is permanently  bidirectional, 
	the Perron vector is therefore constant.
Furthermore, the quantities $a$ and $\alpha$ in~(\ref{eq:aandalpha}) satisfy
	$a \geq 1/\dmax $ and $\alpha \geq 1/( n \, \dmax  )$.
Therefore Corollary~\ref{cor:reversiblerate} takes the form:

\begin{cor}\label{cor:Met}%[\cite{NOOT09}
In any execution of the Metropolis algorithm with a   communication graph~$\dG$ that is permanently  bidirectional, 
	the convergence rate~$\varrho$  satisfies
		 $$\varrho \leq 1-   \min \left ( \frac{2}{\dmax}, \max \left ( 
				\frac{1}{\,  n \, \diamnorm  \,   \dmax } , 
				 \frac{ n }{ \,  \diam \, \bottle \,  \dmax}   \right )  \right )$$ 
	where $\bottle = \bottle(\dG)$, $\diam= \diam(\dG)$, $\diamnorm = \diamnorm(\dG)$, 
	 and $\dmax = \dmax (\dG)$.
\end{cor}

\paragraph{EqualNeighbor algorithm with an Eulerian topology and constant degrees.}
Weights in the \emph{EqualNeighbor} algorithm are given by
	$$ N_{i j}(t) = \left\{ \begin{array}{ll}
	                    \frac{1}{d_i (t)}  & \ \mbox{ if } \ j\in In_i (t)  \\
	                     0 & \ \mbox{ otherwise. } 
	                     \end{array}  \right.$$
If $\dG(t )$ is Eulerian, then the $i$-th entry of the Perron vector of the matrix $N(t)$ is equal to
	$$ \pi_i  (t) =  \frac{d_i(t)}{ |E(t)| } \, ,$$
	where $ |E(t)| = \sum_{i=1}^n d_i(t)$ is the number of edges in $\dG(t)$.
Hence in every execution of the EqualNeighbor algorithm with a   communication graph~$\dG$
	that is permanently Eulerian, the matrices~$N(t)$ share the same Perron vector if
	(a) every directed graph $\dG(t)$ is regular or (b) each node~$i$ has a constant degree~$d_i$.
In case (a), the EqualNeighbor and Metropolis algorithms coincide and Corollary~\ref{cor:Met} applies.
Thus we focus on case (b). 
	
The coefficient $a$ defined in (\ref{eq:aandalpha}) is equal to 
	$$ a = \frac{1}{\dmax} . $$ 
With Corollaries~\ref{cor:Kbcb} and~\ref{cor:KDS},  the bound in Corollary~\ref{cor:reversiblerate} simplifies into:
\begin{cor}\label{cor:EN}
Let $\dG$ be a  dynamic graph that is permanently  Eulerian and such that 
	each node~$i$ has a constant degree~$d_i$.
In any execution of the EqualNeighbor algorithm with the   communication graph~$\dG$ , 
	the convergence rate~$\varrho$  satisfies
	$$  \varrho \leq  1- \min \left ( \frac{2}{\dmax}, \max \left (
				\frac{1}{ \diamnorm \, | E| } , 
				 \frac{ | E | }{   \diam \, \dmax ^{\,2} \, \bottle }	                          \right )            \right ) $$ 
	where $\bottle = \bottle(\dG)$, $\diam= \diam(\dG)$, $\diamnorm = \diamnorm(\dG)$, 
	 $\dmax = \dmax (\dG)$,  and $|E| = \sum_{i \in [n]} d_i$.
%Hence, it holds that $\varrho  \leq  1- \frac{1}{2n^3} $.
\end{cor}

\paragraph{FixedWeight algorithm with an Eulerian topology.}
For each agent $i$, let $q_i$ denote an upper bound on the number of in-neighbors of~$i$
	in a given dynamic graph~$\dG $.
Weights in the FixedWeight algorithm are given by
	$$ W_{i j}(t) = \left\{ \begin{array}{ll}
	                     1/q_i  & \ \mbox{ if } \ j\in \In_i (t) \setminus \{i\} \\
	                     1- \left( d_i(t) -1\right)/q_i & \ \mbox{ if } \ j= i  \\
	                     0 & \ \mbox{ otherwise } \, .
	                     \end{array}  \right.$$
We easily check that if  $\dG(t) $ is Eulerian, then the $i$-th entry of the 
	$W(t)$'s Perron vector  is equal to 
	$$ \pi_i \big(W(t) \big) =  \frac{q_i }{Q} \, ,$$
	where $ Q = \sum_{i \in [n]} q_i$.
It  follows that with a communication  graph that is permanently  Eulerian, the Perron vector  is constant 
	and each matrix~$W(t)$ is reversible.
Furthermore, the quantities $a$ and $\alpha$ in~(\ref{eq:aandalpha}) satisfy
	$a \geq 1/ q $ and $\alpha = 1/Q$.
Using Corollaries~\ref{cor:Kbcb} and~\ref{cor:KDS},  
	 Corollary~\ref{cor:reversiblerate} specializes to the following corollary.
	
\begin{cor}%[\cite{Cha11}
\label{cor:FW}
In any execution of the FixedWeight algorithm with a   communication graph~$\dG$ that is permanently  Eulerian,
	the convergence rate $  \varrho $ satisfies
	$$  \varrho \leq  1-   \min \left ( \frac{2}{q }, \max \left (
				\frac{1}{ \diamnorm \,  Q } , 
				 \frac{ Q }{  \diam \, q ^{ 2} \, \bottle }
	                                       \right ) \right ) $$ 
	where  $\bottle = \bottle(\dG)$, $\diam= \diam(\dG)$, $\diamnorm = \diamnorm(\dG)$, 
	$q  = \max_{i\in [n]} q_i $  and $Q = \sum_{i \in [n]} q_i $.
	%  and  convergence time $ T(\varepsilon)  \leq n^3  \log \big (1/\varepsilon\big)$.
\end{cor}
The quantities $1/\diamnorm  Q $ and $ Q/ \diam \, q ^{ 2} \, \bottle $ in the above bound 
	depend not only on the geometric parameters of~$\dG$, but also on the parameters~$q$ and~$Q$
	of the FixedWeight algorithm,
	and hence cannot be compared in general.
	
\subsection{Quadratic bounds on convergence rates}

Under the conditions specified in Corollaries~\ref{cor:Met} and~\ref{cor:EN}, the convergence rate 
	is bounded above by   $1-1/n^3$ for both the EqualNeighbor and the Metropolis algorithms.
We show that the original Poincar\'e's inequality in Proposition~\ref{prop:specgapDS} yields a convergence rate in $1- O(1/n^2)$
	for Metropolis, and prove that this bound also holds  for EqualNeighbor 
	when the   communication graph  is not too irregular.

First observe that the \emph{Metropolis-length} of any path $\gamma = (i_1, \dots, i_{\ell +1} )$ in $\dG(t)$ 
	of length $ \len {} = \ell$ is given by 
	$$ \len {M(t)} = n \sum_{k \in [\ell] } \max(d_{i_k} (t) ,d_{i_{k+1}} (t) ) ,$$
	while the \emph{EqualNeighbor-length} for a communication graph with constant degrees is
	$$ \len {N (t)}  = |E| \,  \len {} .$$
Our general quadratic bound for Metropolis is based on a simple combinatorial lemma inspired 
	by a nice idea in~\cite{IKOY03}.
	
\begin{lem}\label{lem:diamdeg}
Let $G$ be any bidirectional graph with $n$ nodes, and let  $i_1, \dots, i_{\ell +1}$ be 
	any geodesic in~$G$.
Then 
	$$ \max(d_{i_1},d_{i_2}) + \cdots + \max(d_{i_{\ell }},d_{i_{\ell +1}}) \leq 4n .$$ 
\end{lem}
\begin{proof}
Let ${\cal N}_k$ denote the set of (incoming or outgoing) neighbors of $i_k$, and for each $k \leq \ell $, let 
	$${\cal N}^*_k = \left \{ \begin{array}{ll} 
	                              {\cal N}_k. & \mbox{ if } d_k \geq d_{k+1} \\
	                              {\cal N}_{k +1} & \mbox{ otherwise.}
	                              \end{array} \right . $$
Since $i_1, \dots, i_{\ell +1}$ is a geodesic, ${\cal N}_k$ and $ {\cal N}_{k'} $ are disjoint if 
%	 $0 \leq k < k +3 \leq k' \leq \ell$.
	$k' \geq k+3$.
Hence, ${\cal N}^*_k$ and $ {\cal N}^*_{k'} $ are disjoint if $k' \geq k+4$.
The lemma follows from the pigeonhole principle.                       
\end{proof} 

\begin{prop}\label{prop:quadraticmet}
The Metropolis algorithm with  dynamic  communication graphs that are 
	permanently bidirectional and connected achieves
	 asymptotic consensus with a convergence rate  
	$$  \varrho \leq 1- \frac{1}{4n^2}.$$
\end{prop}	
\begin{proof}
Lemma~\ref{lem:diamdeg} shows that every geodesic $\gamma$ in $\dG(t)$ satisfies
	$$ \len {M(t)} \leq 4 n^2 .$$
Hence,  for every bound $\kDS(M(t))$ defined by (\ref{eq:kds}), it holds that 
	\begin{equation}\label{eq:poinbo}
	\kDS(M(t)) \leq 4 \, b(\dG) .
	\end{equation}
Since $ b(\dG) \leq n^2$, we obtain $\kDS \leq 4 n^2$.
The result  follows from  Corollary~\ref{cor:reversiblerate} and $a \geq 1/n$. 
\end{proof}

The same approach applies to the \emph{Lazy Metropolis} algorithm where weights are defined  by
$$ L_{i j}(t) = \left\{ \begin{array}{ll}
	                     \frac{1}{2 \max \left ( d_i (t)-1, \, d_j(t) - 1 \right)} & \ \mbox{ if } \ j\in  \In_i (t) \setminus \{i\} \\
	                     1- \sum_{ j\in {\cal N}_i (t) \setminus \{i\} } \frac{1}{2 \max \left ( d_i (t) -1, \, d_j(t) - 1 \right) }    & \ \mbox{ if } \ j= i  \\
	                     0 & \ \mbox{ otherwise .} 
	                     \end{array}  \right.$$
Therefore,
	$$ \forall i\in[n],\ \forall t\geq 1, \   \ L_{i i} (t) \geq \frac{1}{2} $$
	and
	$$ \len {L (t)} = 2 n \sum_{k \in [\ell] } \max(d_{i_k} (t) ,d_{i_{k+1}} (t) ) .$$
 Corollary~\ref{cor:reversiblerate} and Lemma~\ref{lem:diamdeg} give the following result for 
	the Lazy Metropolis algorithm.

\begin{prop}\label{prop:quadraticlazymet}
The Lazy Metropolis algorithm with  dynamic  communication graphs that are 
	permanently bidirectional and connected achieves
	 asymptotic consensus with a convergence rate  
	$$  \varrho \leq 1- \frac{1}{8n^2}.$$
\end{prop}	
From the quadratic bound on the hitting times of Metropolis walks proved by Nonaka et al.~\cite{NOSY10},
	Olshevsky~\cite{Ols17} showed that the convergence rate of the Lazy Metropolis algorithm on any
	fixed graph that is connected and bidirectional is bounded from the above by $-1/71 n^{2} $.
Proposition~\ref{prop:quadraticlazymet}   improves  this result  and 
	extends it to the case of  a time-varying topology.

The Metropolis and EqualNeighbor algorithms coincide in the case of communication  graphs that are permanently regular.
Proposition~\ref{prop:quadraticmet} shows that the convergence rate is bounded above by $1- 1/4n^2$ for such 
	topologies, thus extending the quadratic upper bound in~\cite{DS91} for distance transitive graphs to any
	regular graphs.
With moderate irregularity~\cite{Alb97}, a close
	method for bounding $\kDS$ in the EqualNeighbor algorithm gives the following quadratic bound.

\begin{prop}\label{prop:quadraticEN}
In any execution of the EqualNeighbor algorithm with a    communication graph~$\dG$ 
	that is permanently  Eulerian and with  a constant degree~$d_i$  at each node~$i$,  
	 asymptotic consensus is achieved with a convergence rate  
	$$  \varrho \leq 1- \frac{1}{(3 + \dmax - \dmin) \, n^2 }$$
	where $\dmin $ and  $ \dmax$ denote the minimum and maximum degree in each graph $\dG(t)$. 
	\end{prop}	
\begin{proof}
The EqualNeighbor-length of any path in the directed graph~$\dG(t)$ gives
	$$ \kDS(A(t)) =  \frac{1}{| E |} \, \max_{ e }\,  \sum_{ e \in \gamma_{i \, j} }  \lenij {}  d_i \,  d_j   .$$
Hence 
	$$ \kDS(A(t)) \leq \frac{1}{| E |}  \sum_{ i \neq j }  \lenij {}  d_i \,  d_j 
			\leq  \max_{ j \in [n]  }   \sum_{ i \in [n]\setminus  j }  \lenij {}  d_i   .$$
The second inequality is due to the fact that $|E | = \sum_{ k  \in [n]  } d_k $.
An  argument analog to Lemma~\ref{lem:diamdeg} shows that the sum of the degrees along any geodesic is 
	less than $3n$, and thus  each term in the above sum is bounded above by 
	$$ \lenij {}  d_i (t) \leq 3 n + (d_i - \dmin) \lenij {} .$$
The result immediately follows from  Corollary~\ref{cor:reversiblerate} and $a \geq 1/\dmax$. 
\end{proof}

The example of the  \emph{barbell} graph developed by Landau and Odlyzko~\cite{LO81} 
	 shows that the convergence rate of the EqualNeighbor algorithm
	is  greater than $1- 32/n^3$ with a specific set of initial values (see also below).
Thus the general quadratic bound for Metropolis in Proposition~\ref{prop:quadraticmet}  does not hold for EqualNeighbor
	because of degree fluctuations in space.
In the light of this example and of the striking result by Olshevsky and Tsitsiklis~\cite{OT13}, demonstrating
	that the EqualNeighbor algorithm may experience an exponential convergence rate with degree fluctuations in time,
	the Metropolis algorithm appears as a powerful and efficient method  for masking  graph irregularities,
	requiring only bidirectional communication links and limited knowledge at each agent.%~\cite{AKY97}.
	
\section{Bounds for specific communication graphs}\label{sec:examples}

We now examine some typical examples where the bounds presented above 
	are easy to compute.
For the FixedWeight algorithm, we just give the bound derived from the simple geometric bound $\bcb$,
	while we present detailed comparisons of the various bounds for the EqualNeighbor and Metropolis algorithms
	(cf. Figure~\ref{fig:examples}).  
For  Metropolis and FixedWeight, the communication graph is time-varying, but it is supposed to
	belong to one of the listed classes of directed graphs.
In other words, the support is fixed but  node labelling may change over time.
For the EqualNeighbor algorithm, the communication graph is supposed to be fixed  if the directed 
	graphs in the class under consideration are not regular.
This section is completed with the case of the EqualNeighbor algorithm and the fixed \emph{Butterfly} graph,
	which allows us to compare the various methods for bounding convergence rate in the case of  non-reversible
	stochastic matrices.
	
\paragraph{Ring.}
Let $G=(V,E)$ be a bidirectional ring\footnote{%
	For a chain,  graph parameters are of the same order and so leads to bounds of the same order of magnitude.}
	 with an odd number $n = 2m+1$ of nodes.
Here $|E| = 3 n $, $\dmax(G) = 3$, and $\diam(G) = m$.
We easily check that $\diamnorm(G) = m$ and $\bottle(G) = m(m+1)/2$.	
	
Since $G$ is regular, the EqualNeighbor and Metropolis algorithms coincide, and we obtain
	$$ \bcb = 1- \frac{2}{3n^2} + O \left( \frac{1}{n^3} \right), \ 
	    \DS =  1- \frac{16}{3n^2} + O \left( \frac{1}{n^3} \right) . $$
%	    \JS =  1- \frac{8}{9n^2} + O \left( \frac{1}{n^3} \right). 
	
The two bounds are of the same order of magnitude with $ \DS < \bcb $.
Corollary~\ref{cor:EN} gives a convergence rate
	$$ \varrho \leq  1- \frac{16}{3n^2}, $$
	which is the right order for~$n$ large
\paragraph{Hypercube.}  
Let $G=(V,E)$ be the $p$-dimensional cube with $n = 2^p$ nodes.
Here $|E| = (p+1) 2^p $, $d(G) = p +1 $, $\diam(G) = p $, and $\diamnorm(G) = 1 $.
Diaconis and Stroock~\cite{DS91} showed that  $\bottle(G) = 2^{p-1}$.
The EqualNeighbor and Metropolis algorithms coincide, and we obtain
	$$ \bcb = 1- \frac{1}{(p+1) 2^p}  \  \mbox { and } \  
	    \DS =  1- \frac{2}{p \, (p+1)} . $$
%	    \JS =  1- \frac{1}{(p+1)^2} . $$
The  bound $\DS$ is far better than $   \bcb $.
Corollary~\ref{cor:EN} gives a convergence rate  
	$$ \varrho \leq  1- \frac{2}{p \, (p+1)} \leq  1- \frac{1}{(\log_2 n)^2}, $$
	% and thus $ \varrho \leq  1-  O \left( \frac{1}{ (\log n)^2} \right)$.
	 which is the right order for~$n$ large (e.g., see~\cite{Dia88}).
	
\paragraph{Star.}
The star graph with $n$ nodes has $3n-2$ edges.
The maximum degree is $n$, its diameter is~2,  its edge-connectivity is 1, and so its normalized diameter is~$2$.
The bottleneck measure is equal to the lower bound in (\ref{eq:bottleneck}), namely~$n-1$.
 
 For the EqualNeighbor algorithm, we obtain
	$$ \bcb = 1- \frac{1}{ 6n } + O \left( \frac{1}{n^2} \right)  \  \mbox { and } \  
	    \DS =  1- \frac{3}{2n^2} + O \left( \frac{1}{n^3} \right) . $$
%	    \JS =  1- \frac{9}{8n^4} + O \left( \frac{1}{n^5} \right). $$
The bound $\bcb$ is far better then $ \DS  $, and Corollary~\ref{cor:EN} gives a convergence rate
	$$ \varrho \leq  1- \frac{1}{6 n}. $$

As for Metropolis, we have 
	$$ \bcb = 1- \frac{1}{2n^2}  \  \mbox { and } \  
	    \DS =  1- \frac{1}{3(n-1)}  . $$
%	    \JS =  1- \frac{1}{8(n-1)^2} . $$
The bound  $ \DS$  is  asymptotically  better than~$\bcb$ and improves the bound given in~\cite{NOR18}.
Observe that the inequality~(\ref{eq:poinbo}) in the proof of Proposition~\ref{prop:quadraticmet} directly 
	gives $ \varrho \leq  1- 1/4n $.

\paragraph{Two-star.}
A two-star graph $G$ is composed of two identical stars with an edge connecting their centers.
It has an even number $n$ of nodes and  $3n-2$ edges.
Here, $ \dmax (G)= 1 + n/2$, $\diam(G)= \diamnorm(G) =3$, and~$\bottle(G) = n^2/4$.
 
 For the Metropolis  algorithm, we obtain
		$$ \bcb = 1- \frac{2}{ 3 n^2 } + O \left( \frac{1}{n^3} \right)  \  \mbox { and } \  
	    \DS =  1- \frac{8}{3 n^2} + O \left( \frac{1}{n^3} \right) . $$ 
%	    \JS =  1- \frac{8}{n^4} + O \left( \frac{1}{n^5} \right). $$
The bounds  $ \bcb$ and $ \DS  $ are of the same order with  $ \DS < \bcb $.
Corollary~\ref{cor:Met} gives a convergence rate 
	$$ \varrho \leq  1- \frac{8}{3 n^2}. $$
	
As for EqualNeighbor, we have 
	$$ \bcb = 1- \frac{1}{ 9 n } + O \left( \frac{1}{n^2} \right)  \  \mbox { and } \  
	    \DS =  1- \frac{16}{n^3} + O \left( \frac{1}{n^4} \right) . $$
%	    \JS =  1- \frac{288}{n^6} + O \left( \frac{1}{n^7} \right). $$
The bound $\bcb$ is far better than $\DS$, and Corollary~\ref{cor:EN} gives a convergence rate
	$$ \varrho \leq  1- \frac{1}{9n }. $$

\paragraph{Binary tree.}
Consider the full binary tree of depth $p>1$.
It has $n=2^{p +1} -1$ nodes, $3n -2$ edges, and the maximum degree is 4.
The results for the EqualNeighbor and Metropolis algorithms are thus of the same order.
The diameter is $2\, p$ and the normalized diameter is $2\, p$.
We easily check that the bottleneck measure is $2^{p} (2^{p} - 1)$.

For Metropolis, we have
	$$ \bcb = 1- \frac{1}{ 8 \, p (2^{p +1} - 1)}   \  \mbox { and } \  
	    \DS =  1- \frac{ 2^{p +1} - 1 }{ p \, 2^{p +3}  (2^{p } - 1) } . $$ 
%	    \JS =  1- \frac{  (2^{p +1} - 1)^2 }{ p \, 2^{2p +7}  (2^{p } - 1)^2 } . $$ 
The bounds  $ \bcb$ and $ \DS  $ are of the same order with $ \DS < \bcb$.
Corollary~\ref{cor:Met} gives a convergence rate 
	$$ \varrho \leq  1- \frac{1}{2 n \log_2 n }. $$
The results for EqualNeighbor are similar with a convergence rate 
	$$ \varrho \leq  1- \frac{1}{4 n \log_2 n }. $$

Observe that, in the case of a general bidirectional tree, the number of edges remains equal to~$3n -2$ 
	while the diameter may be $n-1$, which leads to 
	$$  \bcb \leq 1- \frac{1}{ 3 \, n^2} $$ 
	for the EqualNeighbor algorithm.
Proposition~\ref{prop:quadraticmet} shows that a quadratic bound also holds for Metropolis.

\paragraph{Two-dimensional grid.}
Let $p$ be an even positive integer, and let $G=(V,E)$ be the two-dimensional grid  with $n = p^2$ nodes.
Here $|E| = p (5p - 4) $, $d(G) = 5 $, $\diam(G) = 2 (p - 1) $,  and $\diamnorm(G) = p - 1 $.
The results for  EqualNeighbor and Metropolis are thus of the same order.
Choosing paths $\gamma_{i,j}$ first with vertical edges and then with horizontal edges yields 
	$\bottle(G) \leq  p^3 (p + 1)/8 $.
	
For the Metropolis algorithm, we obtain
	$$ \bcb = 1 - \frac{1}{ 5 n^{3/2} }   \  \mbox { and } \  
	    \DS \leq  1- \frac{4}{5 p (p-1)(p+1)} \leq 1 - \frac{2}{ 5 n^{3/2} } . $$ 
%	    \JS \leq  1- \frac{8}{25 p^2 (p+1)^2}  \leq 1 - \frac{2}{ 25 n^2}     . $$
The bounds $\bcb$ and $\DS$ are of the same order of magnitude.
Corollary~\ref{cor:Met} gives a convergence rate  
	$$ \varrho \leq  1- \frac{2}{ 5 n^{3/2} }  . $$
Similarly, 	Corollary~\ref{cor:EN} implies that the convergence rate of the EqualNeighbor algorithm satisfies 
	$$ \varrho \leq  1- \frac{4}{ 5 n^{3/2} }  . $$

\paragraph{Barbell.} 
The \emph{barbell graph} $G=( V,E)$ of size $|V| = n = 4p -1$ is composed of two cliques $C$ and $\bar{C}$ 
	with $p$ nodes each,
	that are connected by  a line  of length $2p -1$; see Figure~\ref{fig:barbell}.
The barbell graph is  bidirectional with $|E|= 2p^2 + 6 p -1$ edges. 
The maximum degree is $p+1$.
The diameter and the normalized diameter are equal to $2(p+1)$.
Any geodesic connecting~$i$ to~$j$ with $i \leq 0$ and $j \geq 1$  crosses over the edge $(0,1)$,
	which is thus traversed by $2p(2p -1)$ geodesics.
Clearly $(0,1)$   realizes the maximum in (\ref{eq:bdef}), 
	and hence $\bottle(G) = 2p(2p -1))$.
	
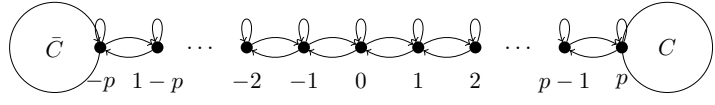
\begin{figure}
\centering
\scalebox{0.8}{
\begin{tikzpicture}[
        dot/.style = {circle, fill, minimum size=0.2cm,
                inner sep=0pt, outer sep=0pt},
        dot/.default = 6pt,  % size of the circle diameter ,
        node_arrow/.style = {->},
        node_text/.style = {}
    ]

    \def \circlewidth {1.5cm}
    \def \nodespacing {0.75cm}
    \def \textspacing {0.2cm}

    \node[draw, circle, minimum size=\circlewidth] (c_bar) {$\bar{C}$} ;

    \node[dot] at (c_bar.east) (_p) {};

    \node[dot, right = \nodespacing of _p ] (1_p) {};

    \node[right = \nodespacing/3 of 1_p ] (etc1) {\dots};

    \node[dot, right = \nodespacing/3 of etc1] (_2) {};
    \node[dot, right = \nodespacing of _2] (_1) {};
    \node[dot, right = \nodespacing of _1] (0) {};
    \node[dot, right = \nodespacing of 0] (1) {};
    \node[dot, right = \nodespacing of 1] (2) {};

    \node[right = \nodespacing/3 of 2 ] (etc2) {\dots};

    \node[dot, right = \nodespacing/3 of etc2] (p_1) {};
    \node[dot, right = \nodespacing of p_1] (p) {};

    \node[draw, circle, minimum size=\circlewidth, right = 0cm of p.center] (c) {$C$} ;

    \node[node_text, below = \textspacing of _p] {$-p$};
    \node[node_text, below = \textspacing of 1_p] {$1-p$};
    \node[node_text, below = \textspacing of _2] {$-2$};
    \node[node_text, below = \textspacing of _1] {$-1$};
    \node[node_text, below = \textspacing of 0] {$0$};
    \node[node_text, below = \textspacing of 1] {$1$};
    \node[node_text, below = \textspacing of 2] {$2$};
    \node[node_text, below = \textspacing of p_1] {$p-1$};
    \node[node_text, below = \textspacing of p] {$p$};

    \drawloop{_p}
    \drawloop{1_p}
    \drawloop{_2}
    \drawloop{_1}
    \drawloop{0}
    \drawloop{1}
    \drawloop{2}
    \drawloop{p_1}
    \drawloop{p}

    \drawarrows{_p}{1_p}
    \drawarrows{_2}{_1}
    \drawarrows{_1}{0}
    \drawarrows{0}{1}
    \drawarrows{1}{2}
    \drawarrows{p_1}{p}

\end{tikzpicture}
}
\caption{The barbell graph}\vspace{-0.0cm}
\label{fig:barbell}
\end{figure}

For the  Metropolis  algorithm, the bounds $\bcb$ and $\DS$ are of the order of magnitude 
	with  $\DS < \bcb$, and $\DS$ is of the order of $1- 32/ p^3 $.
A better estimate on the convergence rate is obtained with~(\ref{eq:poinbo}) and gives
	$$ \varrho \leq 1- \frac{1}{16 p^2} = 1- \frac{1}{(n+1)^2}. $$ 
The barbell graph thus exemplifies  that the bound in Corollary~\ref{cor:Met}  can be far from the original bound 
	$1 - 1/\kDS$.
	
As for EqualNeighbor, the expression of $\kDS$ in~(\ref{eq:kds}) makes the barbell graph as 
	a good candidate for a spectral gap that is cubic in $1/n$.
Indeed, Landau and Odlyzko~\cite{LO81} consider the vector $v \in \IR^n$ defined by 
	$$ v_i = \left\{ 
		\begin{array}{ll}
		-p & \mbox{ if } \  i \in \bar{C} \\
		i & \mbox{ if } \ 1-p \leq i \leq p-1 \\
		p & \mbox{ if } \  i \in C . \\
		\end{array}
	\right. $$
Let $N$ denote the stochastic matrix~associated to the EqualNeighbor algorithm running on the barbell graph.
Proposition~\ref{pro:greengen} shows that
	$$ \Q_N(v) = \frac{1}{2|E|} \sum_{(i,j) \in E} (v_i - v_j)^2  \ \mbox{ and } \ 
	   \lVert v  \lVert_{\pi}^2  =  \frac{1}{|E|}  \sum_{i \in[n]} d_i  v_i ^2  . $$
Hence $$ \Q_N( v ) = \frac{p}{|E|} \ \mbox{ and } \ 
	 \lVert  v  \lVert_{\pi}^2  =  \frac{2}{|E|}  \left( p^4 + \frac{4 p^3}{3} + \frac{p^2}{2}  + \frac{p}{6}  \right) .$$
Therefore
	$$\lambda_2(N) \geq 1- \frac{3}{6p^3 + 8 p^2 + 3 p + 6} \geq 1 - \frac{32}{n^3}.$$
The first inequality is Lemma~\ref{lem:gamma}  and the second one is because $n= 4p-1$.
In the execution with the initial values corresponding to 
	one eigenvector associated  to $\lambda_2(N)$, the convergence rate satisfies
	$$ \varrho = \lambda_2 (N) 
				\geq 1 - \frac{32}{n^3} .$$
Hence, as opposed to the Metropolis algorithm, no general quadratic bound holds for the convergence rate of 
	EqualNeighbor on a fixed connected bidirectional graph.

\paragraph{Butterfly (and EqualNeighbor).}
The \emph{Butterfly graph} has $ n = 2m $ nodes and consists of two isomorphic parts that are connected 
	by a bidirectional edge.
We list the edges between the nodes $1,2 , \dots ,m $ which also determine the edges between the nodes 
	$m+1,m+2, \dots ,2m$ via the isomorphism $\overline{i} = n - i +1$.
The edges between the nodes $1,2, \dots ,m$ are: (a) the edges   $(i+1,i)$  for every $i \in [m-1]$,
	and (b) the edges  $(1,i)$   for every $i \in [m]$. 
In addition, it contains a self-loop at each node and the two edges $(m,\overline{m})$ and $(\overline{m}, m)$. 
Hence, the butterfly graph is not bidirectional but it is strongly connected; see Figure~\ref{fig:butterfly}.

\begin{figure}
\centering
\scalebox{0.8}{
\begin{tikzpicture}[>=latex',tight background]
\node[draw, circle] (n1) at (205.71:2.5) {$1$};
\node[draw, circle] (n2) at (257.14:2.5) {$2$};
\node[draw, circle] (n3) at (308.57:2.5) {$3$};
\node[draw, circle] (n4) at      (0:2.5) {$4$};
\node               (n5) at  (51.43:2.5) {$\dots$};
\node               (n6) at (102.86:2.5) {$\dots$};
\node[draw, circle] (n7) at (154.29:2.5) {$m$};

\draw[->] (n1) -- (n2);
\draw[<-] (n2) -- (n3);
\draw[<-] (n3) -- (n4);
\draw[<-] (n4) -- (n5);
\draw[<-] (n5) -- (n6);
\draw[<-] (n6) -- (n7);
\draw[<-] (n7) -- (n1);

%\draw[<-] (n1) .. controls +(-1.0,-0.5) and +(-1.0,0.5) .. (n1);
\draw[->] (n2) .. controls +(-1.2,0.3) .. (n1);
\draw[<-] (n3) -- (n1);
\draw[<-] (n4) -- (n1);
\draw[<-] (n5) -- (n1);
\draw[<-] (n6) -- (n1);

\begin{scope}[yscale=1,xscale=-1,shift={( 7,0)}]
\node[draw, circle] (m1) at (205.71:2.5) {$\bar1$};
\node[draw, circle] (m2) at (257.14:2.5) {$\bar2$};
\node[draw, circle] (m3) at (308.57:2.5) {$\bar3$};
\node[draw, circle] (m4) at      (0:2.5) {$\bar4$};
\node               (m5) at  (51.43:2.5) {$\dots$};
\node               (m6) at (102.86:2.5) {$\dots$};
\node[draw, circle] (m7) at (154.29:2.5) {$\bar m$};

\draw[->] (m1) -- (m2);
\draw[<-] (m2) -- (m3);
\draw[<-] (m3) -- (m4);
\draw[<-] (m4) -- (m5);
\draw[<-] (m5) -- (m6);
\draw[<-] (m6) -- (m7);
\draw[<-] (m7) -- (m1);

%\draw[<-] (m1) .. controls +(-1.0,-0.5) and +(-1.0,0.5) .. (m1);
\draw[->] (m2) .. controls +(-1.2,0.3) .. (m1);
\draw[<-] (m3) -- (m1);
\draw[<-] (m4) -- (m1);
\draw[<-] (m5) -- (m1);
\draw[<-] (m6) -- (m1);

\end{scope}

% self-loops
\draw[->] (m1) edge [in=30,out=60,min distance=16] node {} (m1);
\draw[->] (m2) edge [in=75,out=105,min distance=16] node {} (m2);
\draw[->] (m3) edge [in=75,out=105,min distance=16] node {} (m3);
\draw[->] (m4) edge [in=85,out=115,min distance=16] node {} (m4);
\draw[->] (m7) edge [in=75,out=105,min distance=16] node {} (m7);

\draw[->] (n1) edge [in=120,out=150,min distance=16] node {} (n1);
\draw[->] (n2) edge [in=75,out=105,min distance=16] node {} (n2);
\draw[->] (n3) edge [in=75,out=105,min distance=16] node {} (n3);
\draw[->] (n4) edge [in=65,out=95,min distance=16] node {} (n4);
\draw[->] (n7) edge [in=75,out=105,min distance=16] node {} (n7);

\draw[->] (n7) .. controls +(-1.3, 0.4) .. (m7);
\draw[<-] (n7) .. controls +(-1.3,-0.4) .. (m7);
\end{tikzpicture}
}
\caption{The butterfly graph}\vspace{-0.0cm}
\label{fig:butterfly}
\end{figure}
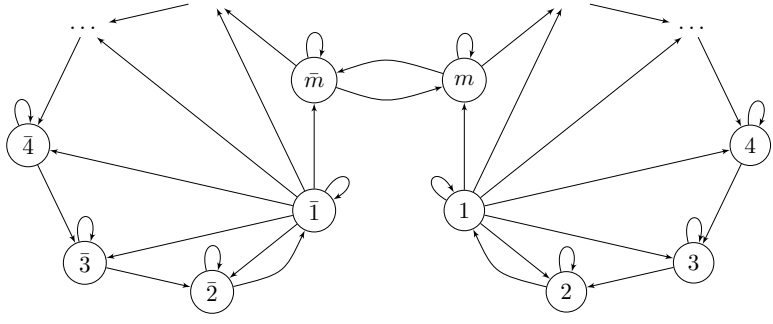

We now consider the EqualNeighbor algorithm running on this fixed graph, yielding a fixed stochastic matrix~$B$
	that is not reversible.
 Corollary~\ref{cor:reversiblerate} is not applicable, but the results 
	in Section~\ref{sec:singular} give a convergence rate
	$$ \varrho \leq 1 -  \max \left ( \frac{\alpha(B)}{n-1} \, , \,  \frac{1}{ \kbcb( B^{\dag}  B ) } \, ,
	                              \,  \frac{1}{ \kDS(B^{\dag}  B)  }  \right )    ,$$
	 where $\alpha(B)$, $\kbcb(  B^{\dag}   B )$, and $\kbcb(  B^{\dag}   B )$ are defined by~(\ref{eq:alpha}), 
	 (\ref{eq:kbcb}), and~(\ref{eq:kds}), respectively.
	 
We easily verify that the Perron vector of $ B$, and thus of $ B^{\dag}   B $,  is given by
	$$\pi_1 = \frac{1}{5}, \ \pi_i =  \frac{3}{5\, . \, 2^i} \  \mbox{ for $i \in \{2, \dots, m-1 \}$ and } \  \frac{1}{5} .$$
By symmetry, this also defines the Perron vector for the remaining indices between $m+1$ and $2m$ since
	$\pi_i = \pi_{n -i +1}$.
Then we easily arrive at
	$$ \alpha ( B) = \pi_m \,  B_{ m \,1} = \frac{1}{5 \, . \, 2^{m-1}} ,$$
	which directly gives the following analytic bound in Proposition~\ref{prop:specgapgen}
	$$  \an = 1- \frac{1}{5 (2m-1) \,  2^{m-1}} .$$
	
For $\kbcb(  B^{\dag}   B )$ and $\kbcb(  B^{\dag}   B )$, we compute the estimates $\bcb( B^{\dag}   B )$
	and $\DS( B^{\dag}   B )$ given by
	$$ \kbcb(  B^{\dag}   B ) \leq \frac{ \alpha(  B^{\dag}   B )}{\diamnorm (H)}
	\ \mbox{ and } \ 
	\kbcb(  B^{\dag}   B ) \leq \frac{ \alpha(  B^{\dag}   B )}{\diam(H)  (\pi_{\max})^2    b(H)} ,$$
	where $H = G_{  B^{\dag}   B} $. 
The directed graph $H$ consists in two  cliques with the sets of nodes $1,2 , \dots ,m$ and 
	$\overline{1}, \overline{2} , \dots ,\overline{m}$ , connected by the  edges 
	$(m-1, \overline{m})$, $(m , \overline{m-1})$, $(m , \overline{m})$ and the three edges in the
	reverse direction.
Thus $H$ has $2(m^2 + 3)$ edges,  $\dmax(H)  = m+2$, $\diam(H) =3$, and  $ \diamnorm (H) = 1$.
The bottleneck measure is $\bottle(H) = m/3$.
A rather tedious computation gives
	$$ \alpha ( B^{\dag}   B ) = \pi_{m-1}  \,  B^{\dag}   B_{ (m -1) \, m} = \frac{1}{3 \, . \, 5 \, .  \, 2^{m-1}} .$$
Since $\pi_{\max} = 1/5$,  we arrive at the two following geometric bounds 
	$$ \bcb = 1- \frac{1}{3 \, . \, 5 \, .  \, 2^{m-1}}   \  \mbox { and } \  
	    \DS =  1-  \frac{5}{3 \, m   \, 2^{m-1}} . $$
The bound $\bcb$ is better than  both $ \an$ and $ \DS  $. 
Thus we arrive at  
	$$ \varrho \leq  1- \frac{1}{3 \, . \, 5 \, .  \, 2^{m-1}} \ . $$

 The subset $S= \{1,2 , \dots ,m \}$ satisfies $\pi(S) = 1/2$ and 
	$$  \sum_{i\in S, j\notin S} \pi_i B^{\dag}   B_{i \, j}  = \frac{1}{5.2^{m-2}} .$$
The lower bound in Cheeger's inequalities gives 	
	$$\lambda_2( B^{\dag}   B) \geq 1 - \frac{1}{5.2^{m-4} } .$$
This lower bound is of the same order as $\bcb$, which shows that  the convergence rate
	of the EqualNeighbor algorithm is  $1- \theta \big(2^{-m}\big)  $.
		
\begin{figure}
$$\begin{array}{|c||c|c|c|}
\hline
\mbox{ }  & \mbox{ EqualNeighbor } & \mbox{ FixedWeight } & \mbox{ Metropolis  } \\
\hline \hline
% \mbox{ complete } & 1 - 1/2n & 1 - n/2q  & 1- 1/ 2n  \\
% \hline 
\mbox{  ring} & 1 - \frac{16}{3n^2} & 1 - \frac{2}{Qn} & 1 - \frac{16}{3n^2}   \rule[-7pt]{0pt}{20pt} \\ 
\hline 
\mbox{ hypercube } & \ \ \  1 - \frac{1}{(\log_2 n)^2}  \ \ \ _{_{[DS]} }  & 1 - \frac{1}{Q} 
                                       & \ \ \ \ \ 1 - \frac{1}{(\log_2 n)^2}  \ \ \ _{_{[DS]} }    \\
\hline 
\mbox{ star } & \ \ \ \ \  1 - \frac{1}{6n}  \ \ \ _{_{[b]} }  & 1 - \frac{1}{2q_{\max}} & \ \ \ \ \ 1 - \frac{1}{3n}  \ \ \ _{_{[DS]} }   \rule[-7pt]{0pt}{20pt} \\ 
\hline
\mbox{ two-star } & \ \ \ \ \  1 - \frac{1}{9 n}  \ \ \ _{_{[b]} }   & 1 - \frac{1}{3 Q} & 1 - \frac{8}{3n^2}   \rule[-7pt]{0pt}{20pt} \\ 
\hline
\mbox{ binary tree} &   1- \frac{1}{4 n \log_2 n } & 1 - \frac{2}{Qn} &  1- \frac{1}{2 n \log_2 n }   \rule[-7pt]{0pt}{20pt} \\ 
\hline 
\mbox{ grid } &  1 - \frac{2}{5n\sqrt{n}}  & 1 - \frac{1}{Q\sqrt{n}}  & 1 - \frac{2}{5n\sqrt{n}} \rule[-7pt]{0pt}{20pt}  \\
\hline 
\mbox{ barbell }   &  1 - \frac{8}{(n+1)^3}       &1 - \frac{2}{Q n}  & 1 - \frac{1}{(n+1)^2}        \rule[-7pt]{0pt}{20pt}  \\
\hline
\end{array}$$
\caption{Bounds for networked systems with $n$ agents and bidirectional  links.}
\label{fig:examples}
\end{figure}

\paragraph{Acknowledgements.} 
I  thank  Eric Fuzy  and Patrick Lambein-Monette for useful discussions, 
	and Jean-Beno\^\i t Bost and Rapha\"el Bost for their help during the completion of this paper.

\bibliographystyle{plain}

\end{document}